\documentclass[preprint,12pt]{elsarticle} 
\usepackage{booktabs}
\usepackage{amssymb}
\usepackage[colorlinks]{hyperref}   
\usepackage{xcolor}
\usepackage{amsthm}
\usepackage{mathtools}
\usepackage{enumitem}
\usepackage{graphicx}
\usepackage{subfigure}
\usepackage{autobreak}
\usepackage{lineno}
\usepackage{natbib}
\usepackage[margin=2cm]{geometry}
\usepackage{multirow}
\usepackage{floatrow}
\usepackage{colortbl}
\usepackage{threeparttable}
\usepackage{subcaption}
\newtheorem{theorem}{Theorem}

\newtheorem{proposition}{Proposition}

\usepackage{graphicx}
\newfloatcommand{capbtabbox}{table}[][\FBwidth]
\usepackage{multirow}
\usepackage{amsthm,amsmath,amssymb}
\usepackage{mathrsfs}
\usepackage{bm}

\usepackage[ruled,linesnumbered]{algorithm2e}

\setcitestyle{authoryear,open={(},close={)}}
\definecolor{orange}{rgb}{1,0.5,0}
\definecolor{red}{RGB}{198,0,35}
\definecolor{amberseldef}{rgb}{1.0, 0.49, 0.0}
\definecolor{ceruleanblue}{rgb}{0.16, 0.32, 0.75}
\definecolor{amber}{rgb}{1.0, 0.49, 0.0}
\definecolor{dodgerblue}{rgb}{0.12, 0.56, 1.0}
\definecolor{pureblue}{rgb}{0, 0, 1.0}

\definecolor{blue}{rgb}{0.0, 0.28, 0.67}

\def\hmath$#1${\texorpdfstring{{\rmfamily\textit{#1}}}{#1}}

\allowdisplaybreaks

\makeatletter\href{}{}
\def\ps@pprintTitle{
   \let\@oddhead\@empty
   \let\@evenhead\@empty
   \let\@oddfoot\@empty
   \let\@evenfoot\@oddfoot
}
\makeatother

\AtBeginDocument{\hypersetup{citecolor= pureblue,linkcolor = pureblue,urlcolor = dodgerblue}}
\begin{document}


\begin{frontmatter}


\address[1]{Department of Civil Engineering, The University of Hong Kong, Hong Kong, China}
\cortext[cor1]{Corresponding authors.}

\author[1]{Wang Chen}
\author[1]{Xinglu Liu\texorpdfstring{\corref{cor1}}{}}
\ead{hsinglul@hku.hk}
\author[1]{Kaihang Zhang}
\author[1]{Hongzheng Shi}
\author[1]{Jintao Ke\texorpdfstring{\corref{cor1}}{}}
\ead{kejintao@hku.hk} 

\title{Joint pricing and matching for dynamic high-capacity ride-sharing considering passengers' choice uncertainty}

\begin{abstract}

This work investigates the uncertainty-aware joint pricing and matching problem for dynamic high-capacity ride-sharing services, where passengers are assumed to be price-elastic and decide whether to accept a ride-sharing offer based on the upfront prices provided by the platform. We formulate the studied problem as a two-stage stochastic program, where the first stage optimizes upfront price decisions for passengers, and the second-stage recourse problem captures passenger-vehicle assignment based on passengers' uncertain choices. To enhance computational efficiency, we introduce a novel relaxation--based gradient descent--guided search algorithm that leverages the problem's structural properties. Initially, the algorithm generates a feasible solution for the first-stage problem via relaxation. It then iteratively improves the solution via a search process guided by the derived gradient information. In particular, scenario reduction is applied to eliminate unnecessary scenarios when calculating the gradient, thereby reducing the overall computational burden. Numerical experiments demonstrate that, compared to solving the stochastic program directly, the proposed algorithm can accelerate computation speed by thousands of times while achieving optimality gaps of no more than 1.1\%. Finally, we validate the benefits of considering passengers’ choice uncertainty through large-scale simulation using real-world datasets and road networks over two large cities. The results demonstrate that, on average, the proposed method can increase the revenue by 5.2\% and the service rate by 8.2\% compared to the baseline approaches. This study provides a valuable reference for transportation network companies to design pricing strategies for ride-sharing to enhance service efficiency and improve revenue.

\end{abstract}


\end{frontmatter}

\newpage

\section{Introduction}

With the proliferation of mobile Internet, ride-sharing has grown and expanded rapidly in recent years~\citep{zhang2021pool, liu2023planning, ke2020pricing, ke2020ride,wang2023optimization}. Ride-sharing allows multiple passengers heading in a similar direction to share their trips and split ride costs, having a high potential to improve traffic efficiency, reduce carbon emissions, and mitigate traffic congestion \citep{santi2014quantifying, shaheen2015shared, chen2024development, chen2024quantifying, chen2025scaling, hua2022optimality, najmi2017novel, stiglic2016making}. Many transportation network companies (TNCs) have launched ride-sharing services (as shown in Figure \ref{fig: TNCs}) to improve vehicle utilization and build a sustainable mobility system. For example, in 2023 Quarter 4, Uber expanded ride-sharing services to 11 new markets, and shared rides crossed \$1 billion in annualized bookings \citep{Uber2023Q4}. Also, in 2022, the users of the DiDi ride-sharing service shared a total of 1.69 billion kilometers in rides \citep{Didi2022}. In addition, with the maturity of autonomous driving (AD) technology, a few high-tech companies have launched autonomous mobility-on-demand services. For instance, Waymo provides autonomous ride-hailing services in Phoenix, San Francisco, Los Angeles, and Austin \citep{Waymo2024}. These driverless vehicles are fully compliant and can be operated to carry multiple passengers on each trip, which offers a promising way to promote ride-sharing services drastically.

\begin{figure}[!ht]
    \centering
    \includegraphics[width = 0.8\textwidth]{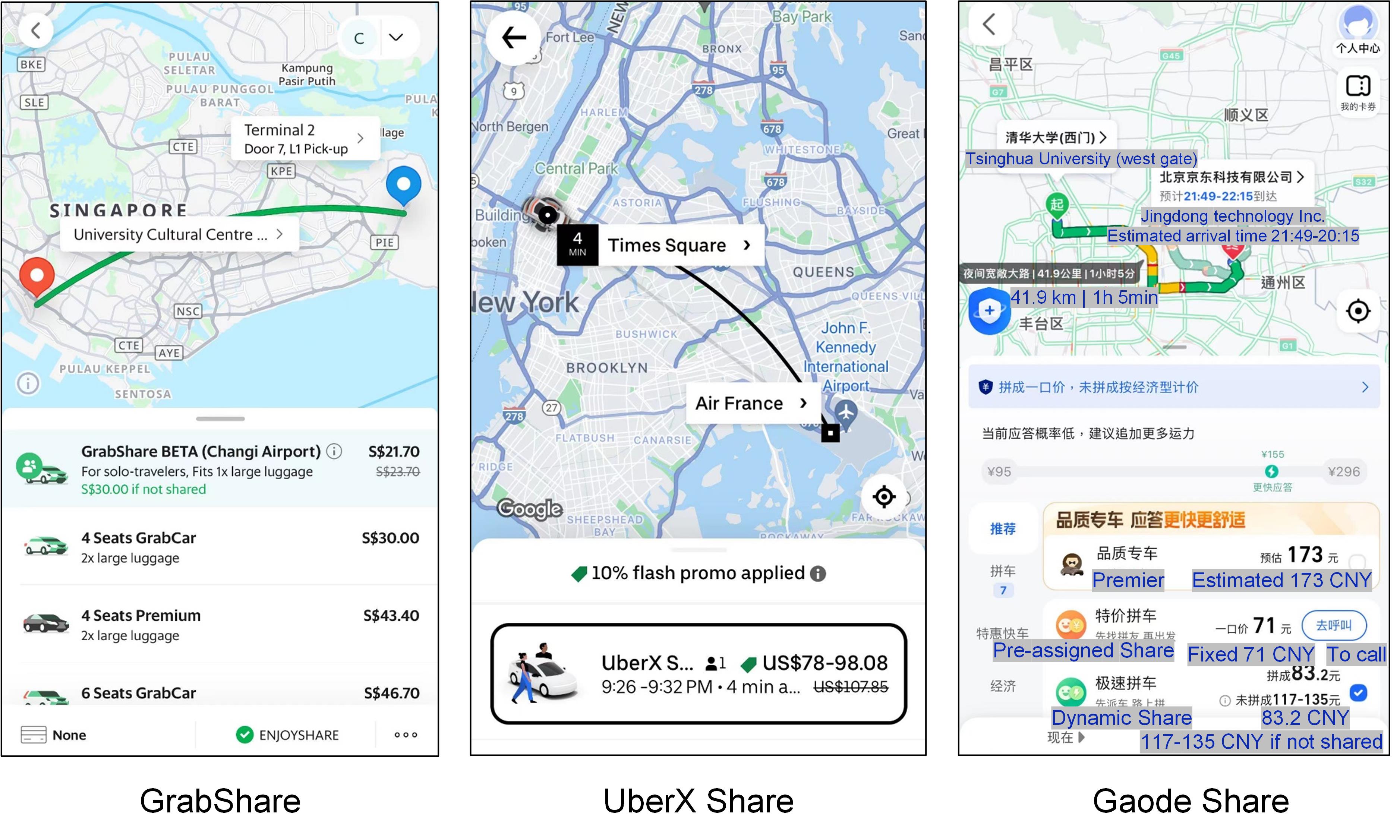}
    \caption{Examples of transportation network companies (TNCs) providing ride-sharing services include Grab, Uber, and Gaode (a leading Chinese map company integrating ride-hailing services from many TNCs).}
    \label{fig: TNCs}
\end{figure}

A ride-sharing platform attracts passengers by offering discounts, \emph{i.e.}, a lower price than solo-hailing. Specifically, when passengers send their mobility requests to the platform, the platform needs to offer upfront prices that will be charged to passengers, as shown in Figure \ref{fig: TNCs}. Then, passengers decide whether to accept the upfront prices. If they accept ride-sharing, the platform will match them with available vehicles. If they decline the offer, they will choose other transportation modes. In this context, the ride-sharing platform must consider the uncertainty of passengers' choices. In particular, if the platform offers high upfront prices for passengers, they may not accept ride-sharing services, leading to a low vehicle utilization rate as well as low revenue. However, if the platform charges low upfront prices to passengers, more passengers are willing to participate in ride-sharing, but low trip fares from passengers may also lead to revenue loss. Hence, passengers' choice uncertainty should be carefully considered when designing upfront prices for them, which can increase the platform's revenue and enhance fleet utilization.

In addition, the platform should consider passenger-vehicle assignment results when designing the upfront pricing strategy, which can further improve vehicle utilization and operational efficiency. Specifically, if multiple passengers with similar itineraries can share their trips and be assigned to one vehicle, the platform should offer a low upfront price (\emph{i.e.}, a high discount) to attract these passengers to ride-sharing, improving service efficiency and overall revenue. Also, when the demand-to-supply ratio is low, the platform should offer a relatively high discount to attract more passengers, improving the overall utilization of the fleet. The pricing and matching problem in high-capacity ride-sharing scenarios can be dramatically more complicated because the platform needs to check the shareability among multiple passengers and consider the uncertainty of their choices, which significantly increases the computational costs.

The uncertainty of passengers' choices and the interaction between pricing and matching make it difficult to design reasonable upfront prices for passengers in ride-sharing services. Existing studies have proposed various algorithms to optimize upfront pricing strategies of ride-sharing services \citep{jacob2021ride, shi2024second, wang2024spatial, liu2024pricing}. For example, \cite{wang2024spatial} adopted reinforcement learning (RL) to optimize spatial-temporal upfront prices of ride-sharing, which could improve the system performance. However, most existing methods overlooked the interaction between upfront pricing and passenger-vehicle assignment. In addition, the methods proposed by existing studies cannot address passengers' choice uncertainty efficiently. Also, these methods cannot be extended to high-capacity scenarios, where the possibilities for matching vehicles and passengers can be numerous, which cannot be effectively addressed through enumeration.

To address the research gaps mentioned above, this study proposes a joint optimization framework for dynamic high-capacity ride-sharing considering the uncertainty of passengers' choices. The main contributions of this study are summarized as follows:

\begin{itemize}
    \item This study proposes an uncertainty-aware joint optimization framework. Considering passengers' choice uncertainty, a two-stage stochastic program is adopted to optimize upfront prices and vehicle-passenger assignment simultaneously. In particular, the first stage determines the optimal upfront prices for passengers considering their different choices, and the second stage matches passengers with vehicles to maximize the overall utility. This uncertainty-aware joint optimization method can take into account the interaction of upfront pricing and vehicle-passenger assignment.

    \item This study proposes a novel relaxation--based gradient descent--guided search (RGDS) algorithm based on the problem's structural properties. It leverages a simple provable property to generate a feasible solution for the first-stage problem (\emph{i.e.}, upfront prices) on short notice, and then iteratively improves the solution via a search process guided by the gradient descent directions. We also incorporate an optimality-guaranteed scenario reduction component to shrink the scenario set and accelerate computation.  Our numerical experiments demonstrate that the RGDS algorithm can dramatically reduce computational costs while achieving high accuracy, accelerating computation speed over one thousand times with a gap of no more than 1.1\%.

    \item This study conducts extensive numerical experiments in Chengdu and Shanghai on an agent-based simulation platform calibrated with real-world mobility data and road networks. The experimental results demonstrate that the proposed joint optimization method can achieve a noticeable improvement in the platform's revenue and passengers' service rate across all scenarios, increasing the revenue on average by 5.2\% and service rate by 8.2\%, respectively.
\end{itemize}

The remaining parts of this paper are organized as follows. First, the various optimization methods proposed by recent studies for ride-sharing are reviewed and discussed in Section \ref{sec: LR}. Subsequently, the problem description and joint optimization framework formulation are introduced in Section \ref{sec: prob_desc}. Then, the proposed RGDS algorithm is explicitly explained in Section \ref{sec: relaxation}. In Section \ref{sec: compare_problems}, numerical studies are conducted to compare the computational performance among different formulations. Moreover, in Section \ref{sec: exp}, experiments in two real-world cases involving different vehicle capacities and fleet sizes are conducted, proving the effectiveness of the proposed optimization framework. Finally, conclusions and future research directions are summarized in Section \ref{sec: conclusion}.

\section{Literature review}\label{sec: LR}

In this section, we review the related studies in recent years. First, various algorithms proposed to solve the vehicle-passenger assignment problem in ride-sharing systems are discussed in Section~\ref{sec: ride_sharing_alg}. In addition, the pricing and matching optimization methods and joint optimization frameworks of ride-sharing services are reviewed in Sections~\ref{sec: single_optimization} and~\ref{sec: joint_optimization}, respectively. Finally, the research gaps are concluded in Section~\ref{sec: remark}.

\subsection{Ride-sharing algorithms}
\label{sec: ride_sharing_alg}

Since dynamic ride-sharing is extremely complicated, a lot of studies have proposed various algorithms to improve the efficiency of dynamic ride-sharing. \cite{di2013optimization} proposed a dynamic ride-sharing model that could represent the system performance, based on which the best assignment between passengers and vehicles could be found in the transportation road network using an optimization program. However, the optimization problem could be extremely time-consuming when the road network is large. \cite{santi2014quantifying} introduced shareability networks that translated ride-sharing problems into a graph-theoretic framework where trips that could potentially be shared were connected based on time constraints. Based on the shareability network method, \cite{kucharski2020exact} proposed a utility-based formulation that enabled cross-scenario sensitivity analysis, including pricing and regulatory strategies. This method offered efficient solutions but assumed all trips were known in advance. Furthermore, \cite{alonso2017demand} introduced an on-demand, high-capacity dynamic ride-sharing algorithm that can solve large-scale dynamic ride-sharing problems and has been widely used to address various ride-sharing challenges. In addition, the passenger-vehicle assignment has been optimized by various algorithms \citep{bei2018algorithms, tafreshian2020trip, sundt2021heuristics}. However, these studies focused on improving ride-sharing algorithm efficiency, \emph{i.e.}, reducing computational costs, rather than improving the performance of dynamic ride-sharing services, \emph{e.g.}, increasing the revenue of the platform or enhancing the service rate of passengers.

\subsection{Pricing and matching optimization}
\label{sec: single_optimization}

Various optimization methods in terms of pricing or matching have been proposed recently to improve the performance of ride-sharing \citep{shi2024second, liu2024two, SCHWARZSTEIN2023106094}. For example, \cite{shi2024second} proposed a second-pricing-based ride-sharing mechanism to improve vehicles' willingness to participate in the ride-sharing system and achieve greater social welfare. \cite{liu2024two} proposed a two-stage matching method for the high-capacity ride-sharing problem, using a distance matrix and sliding time windows to represent the distance and time constraints of multiple riders, which could improve computational efficiency. In addition, a few studies tried to optimize ride-sharing using RL algorithms \citep{yu2019integrated, al2019deeppool, shah2020neural, si2024vehicle, wang2024spatial, GAO2024106550}. For example, \cite{shah2020neural} adopted an RL model to predict the future utility of each assignment and integrated it into the current matching process to optimize the long-term performance of dynamic ride-sharing.

Also, researchers adopted queuing models to capture passengers' behavior and vehicles' state transitions in a ride-sharing system, providing insight into regulating and optimizing dynamic ride-sharing. For example, \cite{jacob2021ride} developed a queuing model to design the price-service menu to maximize the ride-sharing platform's revenue. In addition, \cite{daganzo2019general} proposed a queuing network modeling framework that can connect the fleet size to passengers’ expected travel experience in a steady state. This framework has been adopted by a few studies to address ride-sharing challenges, such as guaranteeing passengers' detour and waiting time \citep{daganzo2020analysis, ouyang2021performance} and integrating ride-sharing services into the public transit network \citep{liu2021mobility}. Furthermore, \cite{liu2023planning} extended the spatial queuing network model to a multi-zone ride-sharing version that can describe how vehicles' state transitions within and across multiple zones, and how passengers are served by idle or partially occupied vehicles in a steady state.

These studies focused on only one perspective of ride-sharing without considering its interaction. For example, studies that aimed to optimize matching for ride-sharing just assumed a given demand without considering pricing strategies. However, in real-world ride-sharing systems, demand is price-elastic and can significantly impact ride-sharing performance, which should not be overlooked when designing a strategy. Instead, upfront pricing and matching should be taken into account simultaneously when optimizing a ride-sharing system.

\subsection{Joint optimization
for ride-sharing services}
\label{sec: joint_optimization}

A few studies aimed to jointly optimize ride-sharing services through various methods~\citep{shah2022joint, xianjie2023AAAI, haliem2021distributed, jiao2022incentivizing, wang2025optimizing, TUNCEL2023106317, XU2025106998, PANTUSO2022105802}. For example, \cite{shah2022joint} first optimized passenger prices by considering ride-sharing as an auction, where the buyers are the customers who want rides, and the seller is the ride-sharing operator. Then, the platform matched the passengers who accepted the prices and available vehicles, intending to maximize the revenue. Similarly, \cite{xianjie2023AAAI} proposed a two-layer RL model to address this issue, where the first layer was responsible for computing candidate prices for transportation requests and the second layer could estimate future values for different possible matches. Thus, passengers and vehicles could be matched through an integer optimization program based on the prices and future values, which is similar to the matching process proposed by \cite{shah2020neural}. In addition, \cite{haliem2021distributed} proposed a dynamic, demand-aware, and pricing-based vehicle-passenger matching framework for ride-sharing. The authors first determined the prices for passengers based on the predicted demand, then matched passengers with vehicles through the greedy assignment. However, these studies optimized pricing and matching sequentially without considering their interaction. \cite{jiao2022incentivizing} proposed a dynamic discount pricing strategy for the platform, which can be integrated into the matching algorithm to improve the platform's profit. However, this study focused on pre-assigned low-capacity ride-sharing, where at most two passengers can be pooled on each trip, and they can be assigned to a vehicle only if both passengers accept ride-sharing services.

\subsection{Remark}
\label{sec: remark}

Most of the studies mentioned above did not consider the interaction between upfront pricing and passenger-vehicle assignment when optimizing one perspective of ride-sharing. Also, they overlooked the passengers' choice uncertainty, \emph{i.e.}, whether passengers are willing to accept ride-sharing services, which can significantly impact the system performance. In addition, these studies primarily focused on low-capacity scenarios, allowing at most two passengers to share a trip. However, in real-world applications, especially in the AD era, a vehicle such as a minibus or robotaxi can carry multiple passengers on each trip. Also, these vehicles can be assigned new passengers en route, without requiring all passengers to accept ride-sharing at each matching step. Therefore, this study aims to jointly optimize upfront pricing and matching for dynamic high-capacity ride-sharing considering passengers' choice uncertainty. In particular, this study considers passengers to be price-elastic and calibrates their elasticity based on a questionnaire survey. Furthermore, this study proposes a joint optimization method that can be modeled as a two-stage stochastic program, and designs a relaxation--based gradient descent--guided search (RGDS) algorithm to solve the program efficiently. The numerical experiments demonstrate that the proposed optimization framework can significantly improve the performance of dynamic high-capacity ride-sharing.

\section{Problem description}\label{sec: prob_desc}

\subsection{Basic setting}\label{sec: basic_setting}

This study considers a dynamic ride-sharing platform operated by a TNC in an urban area with a fleet of vehicles. Let $K$ denote the fleet, and the fleet size is denoted as $\lvert K \rvert$. In addition, the maximum capacity of vehicles is $C$, \emph{i.e.}, the platform can assign at most $C$ passengers to a vehicle on each trip. This study assumes the platform provides only dynamic ride-sharing services with a fixed number of fully compliant vehicles. This is reasonable in the AD era because all vehicles and passengers can be operated by a central platform, dramatically increasing the potential of ride-sharing. In addition, autonomous vehicles can replenish their energy (such as charging) when they run out of it, while they can continue to work at other times. Hence, the fleet size can be the same across the experiment horizon. As mentioned before, when passengers submit transportation requests to the platform, the platform charges upfront prices to passengers according to their pick-up and drop-off locations, current demand and supply, and so on. Passengers then determine whether to accept the upfront prices or not. If not, passengers opt for other platforms or public transit. Otherwise, the platform matches passengers with vehicles according to the batch-matching discipline that evaluates a group of passengers and vehicles and tries to optimize the overall utility for the entire group. It should be noted that this study assumes each request consists of only one passenger.

At each decision step, let $R := \{R_w, R_n\}$ denote the transportation requests, where $R_w$ denotes awaiting passengers who have accepted the offered upfront prices in previous steps but have not been assigned to any vehicles, and $R_n$ represents new requests. Compared with solo-hailing services, the platform must offer reasonable upfront discounts to attract new requests. Let $\boldsymbol{\theta} := \left( \theta_r \right)_{r \in R_n}$ denote the offered discounts. Regarding awaiting passengers, since they have accepted previously offered discounts, the platform does not need to charge them again. This study assumes that passengers will abandon their requests after waiting for a while (e.g., 10 minutes) before being matched with a vehicle. However, they will not cancel their requests after they are matched with vehicles. Since the platform provides only ride-sharing services, passengers arrive on the platform with an expected minimum discount for their ride. In other words, passengers will accept the upfront price offered by the platform if and only if the offered discount is higher than their expected value. Those values are assumed to be independent and identically distributed among all passengers and are associated with the cumulative distribution function (CDF) $\hat{G}$. Hence, the probability for passengers to accept ride-sharing can be calculated using the CDF as $G(\theta_r) = 1 - \hat{G}(\theta_r), \forall r \in R_n$.

To calibrate the CDF of passengers’ acceptance behavior, we conduct a questionnaire survey that elicits the minimum discount required for passengers to accept ride-sharing services under different detour times. As a result, the acceptance probability function \( G(\theta_r) \) depends not only on the offered discount \( \theta_r \) but also on the associated detour level. To simplify the simulation and reduce computational complexity, we focus on the case with a 40\% detour time in this study. Based on the survey responses at this detour level, we approximate passengers’ choice probability using a quadratic function, which provides a smooth and tractable representation of the empirical acceptance behavior. The fitted acceptance function is given by
\begin{equation}\label{eq:G}
    G(\theta_r) = -0.5017\, \theta_r^2 + 0.9900\, \theta_r - 0.0544,
\end{equation}
where \( \theta_r \in [\theta^L, \theta^U] \). The high goodness of fit indicates that the quadratic approximation captures passengers’ acceptance patterns well. The detailed survey design, data screening procedure, and estimation methodology are presented in \ref{seca: survey}.
It should be noted that our model captures both the offered discount and the experienced detour through a discrete guarantee-level structure. Specifically, the platform pre-commits to a maximum detour guarantee (e.g., 40\%) when offering the discount $\theta_r$. The calibrated function $G(\theta_r)$ in Equation~\eqref{eq:G} therefore represents the acceptance probability conditional on this guaranteed detour bound, estimated from survey responses elicited at the corresponding detour level. In the second-stage matching and routing, the actual realized detour for any pooled trip never exceeds this pre-specified guarantee, ensuring that passengers decide under consistent information. This approach balances behavioral realism with computational tractability by discretizing the continuous detour space into a finite set of guarantee levels.

For a ride request, the trip fare consists of the basic fare and a variable fare proportional to the trip distance. Let $\eta_1$ and $\eta_2$ denote the basic and proportional trip fare, respectively. Hence, if a passenger with a trip distance $l_r$ accepts the offered discount, then the ride-sharing trip fare (denoted as $p_r$) can be calculated as follows:
\begin{equation}\label{eq: p_r}
    p_r = (1-\theta_r)(\eta_1 + \eta_2 l_r), \quad \forall r \in R_n.
\end{equation}
Remind that the trip fares for the awaiting passengers, $p_r, \forall r \in R_w$, have been appointed in previous steps. Intuitively, since the awaiting passengers have accepted ride-sharing, the platform can offer fewer discounts for new passengers if many passengers are already waiting to be matched. On the contrary, if only a few passengers are waiting for assignment, the platform can offer increased discounts to attract more passengers to share trips.

\subsection{Operation procedure}\label{sec: operation}

Figure \ref{fig: operation} illustrates the platform's operation procedure. Specifically, at each decision step, the platform first determines which passengers and vehicles can be potentially matched based on the matching radius; that is, passengers can be assigned to vehicles in their matching areas only. As shown in Figure \ref{fig: operation} (a), two vehicles, $k_1$ idle and $k_2$ with one passenger onboard $r_0$, are located in the matching areas of two passengers, $r_1$ and $r_2$. Hence, $r_1$ and $r_2$ can be potentially assigned to $k_1$ and $k_2$. While passenger $r_3$ can only be potentially matched with vehicle $k_2$ since $k_1$ is out of the passenger's matching area. Subsequently, for the candidate passengers of each vehicle, the platform determines which passengers can be pooled based on the maximum detour constraint. Specifically, the platform designs the shortest path for each vehicle to sequentially visit the passengers and calculates the detour time for each passenger. If the detour time for all passengers is no more than the restricted maximum detour time, then the passengers can be pooled together as a trip. For example, passengers $r_1$ and $r_2$ can share their trips and be pooled as one trip. However, they cannot share a trip with $r_0$ due to the long detour distance. As a result, they cannot be assigned to vehicle $k_2$. Passenger $r_3$ can be matched with $k_2$ as the itineraries of $r_0$ and $r_3$ are similar. In high-capacity scenarios, calculating the shortest paths for vehicles can be time-consuming. Hence, this study adopts a heuristic algorithm that lets vehicles sequentially visit their nearest spots to obtain the shortest path, achieving a high accuracy but with significantly lower computational costs compared with enumeration \citep{chen2024quantifying}. This enables the proposed optimization method to be applied to high-capacity scenarios with reasonable computation time.

Subsequently, as shown in Figure \ref{fig: operation} (b), the platform calculates the utility for matching vehicles with trips. The utility for trips consists of pickup costs and trip fares. This study considers only pickup costs since the delivery cost can be explained by the proportional fare of a request. Let $F$ denote all trips in each decision step. Let $\delta$ denote the cost of unit pickup distance, then the utility for assigning a trip to a vehicle can be calculated as follows:
\begin{equation}
    \label{eq: utility}
    u_{kf} = \sum_{r \in f}{(p_r - \delta \cdot l_{kr})}, \quad \forall k \in K, f \in F,
\end{equation}
where $l_{kr}$ denotes the pickup distance for vehicle $k$ to pick up request $r$. Remind that $p_r$ denotes the discounted trip fare of request $r$. Then, the platform determines which trips and vehicles are matched by maximizing the overall utility. As shown in Figure \ref{fig: operation} (c), the platform assigns trip 1 to $k_1$ by pooling passengers $r_1$ and $r_2$ and trip 4 to $k_2$. Finally, vehicles travel to pick up and deliver their scheduled passengers according to the designed shortest paths, as shown in Figure \ref{fig: operation} (d).

\begin{figure}[!ht]
    \centering
    \includegraphics[width = 0.7 \textwidth]{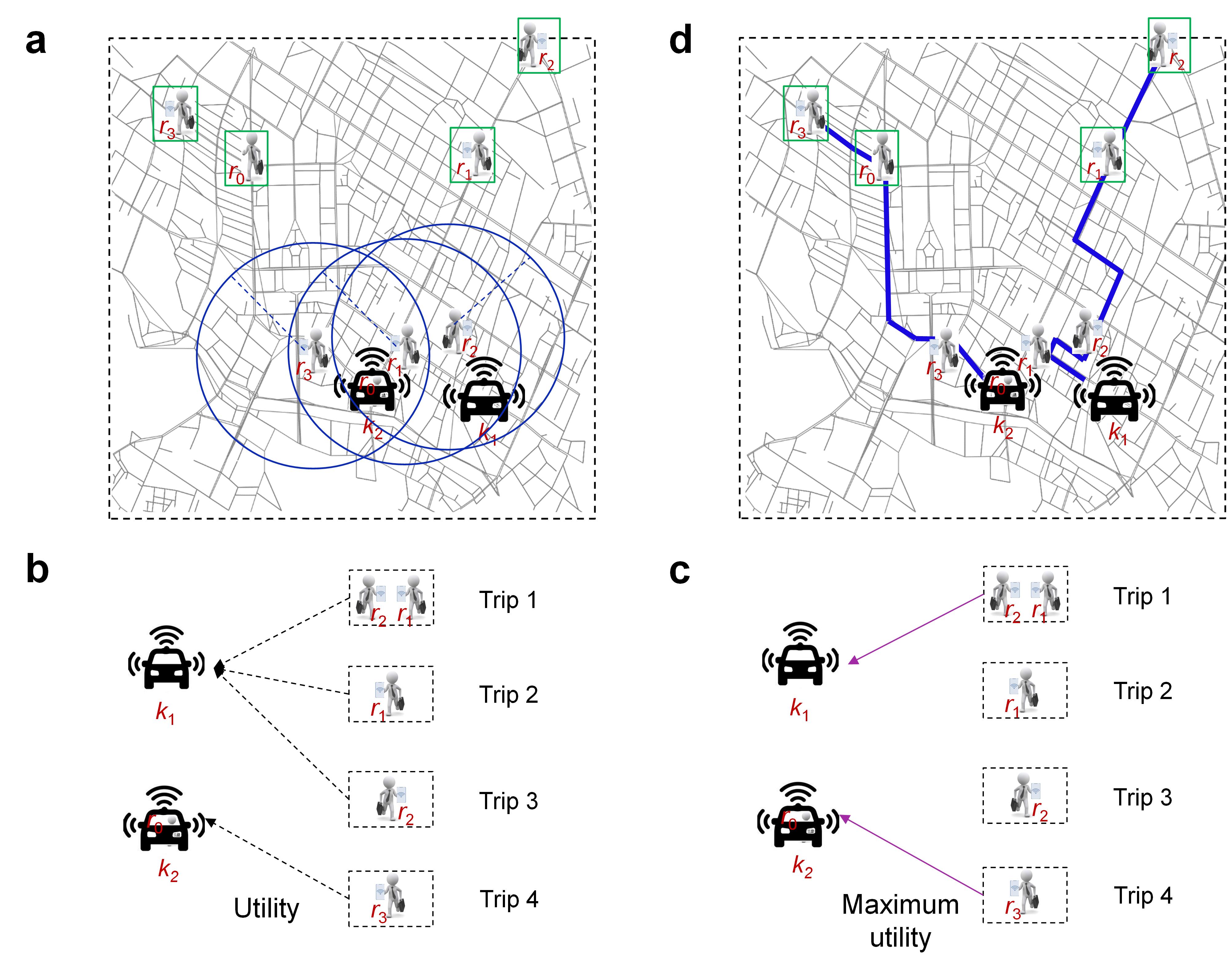}
    \caption{Illustration for the operation procedure. Blue circles and green boxes denote passengers' matching areas and destinations, respectively. Blue lines represent vehicles' routes.}
    \label{fig: operation}
\end{figure}

The vehicles are assumed to fully comply with the platform's operation, always accepting the scheduled trips and traveling via the designed routes. In addition, this study does not consider the effect of traffic congestion and assumes all the vehicles travel at a speed of $v$. It should be noted that partially occupied vehicles can be assigned new passengers en route since this study focuses on dynamic ride-sharing. Once a vehicle has been assigned a new trip, the platform updates the route of the vehicle; otherwise, the vehicle will travel according to the original itinerary. We refer readers to \cite{chen2025hrsim} for more operation details.

\subsection{A two-stage stochastic programming formulation}\label{sec: stochastic_prog}

Given a group of discounts $\bm{\theta}$, the distribution of new passengers' decisions can be represented as $\Xi = \{ (y_r)_{r \in R_n}: y_r \in \{0,1\}, \forall r \in R_n \}$. If a new passenger $r$ accepts the upfront price, then $y_r = 1$; otherwise $y_r = 0$. Since this study assumes that passengers' choices are independent and identically distributed, given a discount $\theta_r$ for a passenger $r$, the probability for the passenger to accept ride-sharing or not can be represented as follows:
\begin{equation}\label{eq: y_r}
    y_r =   \left\{
                \begin{array}{cl}
                    1, & \text{with probability } G(\theta_r);  \\
                    0, & \text{with probability } 1 - G(\theta_r). 
                \end{array}
            \right.
\end{equation}
Furthermore, let $| R_n |$ and $| R_w |$ denote the number of new and awaiting passengers at each decision step, respectively. Hence, there are $| \Xi | = 2^{| R_n |}$ scenarios in total, as each new passenger can decide whether to participate in the ride-sharing service. The probability of each scenario can be calculated based on the distribution of $y_r$, as follows:
\begin{equation}\label{eq: P_xi}
    P(\xi) = \prod_{r \in R_n} \left[G(\theta_r)\right] ^ {y_r^\xi} \left[1-G(\theta_r)\right]^{1-y_r^\xi}, \quad \forall \xi \in \Xi,
\end{equation}
where $y_r^\xi$ denotes the acceptance decision of passenger $r$ in scenario $\xi$. $P(\xi)$ denotes the probability of scenario $\xi$ , and $\sum_{\xi \in \Xi} P(\xi) = 1$.

To consider passengers' stochastic behavior when determining upfront prices for passengers, this study models this problem as a two-stage stochastic program (SP). This first stage aims to determine the optimal discounts for new passengers to maximize the expected utility. The second stage calculates the entire utility in each scenario through ILP. Specifically, in each scenario $\xi$ ($\forall \xi \in \Xi$), the platform assigns feasible trips $F^\xi$ to vehicles $K$ by determining the binary variables $x_{kf}^\xi$ ($\forall f \in F^\xi, k\in K$) to maximize the overall utility (we write the objective in minimization form, i.e., minimize the negative utility). We let $x_{kf}^\xi = 1$ denote that trip $f$ is scheduled to vehicle $k$ in scenario $\xi$; otherwise, $x_{kf}^\xi = 0$. Let $u_{kf}^\xi$ denote the utility associated with assigning trip $f$ to vehicle $k$ in scenario $\xi$.

The two-stage stochastic program is as follows:
\begin{align}
    \label{eq: SP}
    O = \min_{\bm{\theta}} & \,\,\,\,\mathbb{E}_{\xi \in \Xi}\left[- Q(\bm{\theta}, \xi)\right] && \textbf{(SP)}
   \\ 
\text{s.t.} \quad
    & \theta_r \in [\theta^L, \theta^U],  && \forall \theta_r \in \bm{\theta}.
    \label{eq:c1}
\end{align}
The second-stage problem is given by:
\begin{align}
    Q(\bm{\theta}, \xi) = \min_{x^{\xi}}& \sum_{f \in F^\xi}\sum_{k\in K}{-x_{kf}^\xi\cdot u_{kf}^\xi}&&
\\
\text{s.t.} \quad
    \label{con: SP_C1}
    &\sum_{f \in F^\xi}{x_{kf}^\xi \leq 1}, &&\forall k\in K,
\\
    \label{con: SP_C2}
    &\sum_{k\in K}\sum_{f \in F^\xi:r\in f}{x_{kf}^\xi \leq y_r^\xi}, \quad &&\forall r\in R_n,
\\
    \label{con: SP_C3}
    &\sum_{k\in K}\sum_{f \in F^\xi: r \in f}{x_{kf}^\xi \leq 1}, \quad &&\forall r\in R_w,
\\
    &x_{kf}^\xi \in \{ 0,1 \}, \quad &&\forall f \in F^\xi, k\in K.
\end{align}

Objective function \eqref{eq: SP} minimizes the negative system-wide expected utility. Constraint \eqref{eq:c1} defines the feasible range of the pricing decisions.
Constraints \eqref{con: SP_C1} ensure each vehicle can be assigned at most one trip. Constraints \eqref{con: SP_C2} ensure that a new passenger can be matched with a vehicle only if they have accepted the upfront prices, and constraints \eqref{con: SP_C3} ensure a waiting passenger can be matched with at most one vehicle.

Solving the above stochastic optimization problem can be time-consuming. On the one hand, the problem is strongly NP-hard, as stated in Theorem \ref{theo: NP_hard}. On the other hand, as listed in Table \ref{tab: size_problem}, the master problem of the SP consists of $\mathcal{O}(|R_n|)$ constraints and $\mathcal{O}(|R_n|)$ continuous variables, and the second-stage problem in each scenario consists of $\mathcal{O}(|K||F|)$ constraints and $\mathcal{O}(|K||F|)$ integer variables. The second-stage problem is already complicated, and the number of scenarios increases exponentially as more new passengers join the platform. Consequently, this problem cannot be solved efficiently. However, the platform must solve this problem at each decision step in real-world applications, offering upfront prices for new passengers on short notice. Therefore, this study proposes a novel relaxation--based gradient descent--guided search algorithm to obtain a solution with comparable accuracy but significantly lower computational costs.


\begin{theorem}\label{theo: NP_hard}
    The two-stage stochastic optimization problem is strongly NP-hard.
\end{theorem}

\begingroup
\setlength{\tabcolsep}{6pt} 
\renewcommand{\arraystretch}{1} 
\begin{table}[!ht]
\caption{Sizes of different problems.}
\begin{center}
\begin{tabular}{ l l l l}
\hline
Problem &  \# constraints & \# integer variable & \# continuous variable \\
\hline
SP first stage& $\mathcal{O}(|R_n|)$   & -- & $\mathcal{O}(|R_n|)$\\
SP second stage& $\mathcal{O}(|K||F|)$  &  $\mathcal{O}(|K||F|)$ & -- \\
RP      & $\mathcal{O}(|K||F|)$  & -- & $\mathcal{O}(|K||F|)$ \\
RP-CO   & $\mathcal{O}(C|K||F|)$ & -- & $\mathcal{O}(C|K||F|)$ \\
RAP-AS  & $\mathcal{O}(|\Xi||K||F|)$ & -- & $\mathcal{O}(|\Xi||K||F|)$ \\
RAP-RS  & $\mathcal{O}(|\Xi^{\text{RS}}||K||F|)$ & -- & $\mathcal{O}(|\Xi^{\text{RS}}||K||F|)$ \\

\hline

\end{tabular}
\label{tab: size_problem}

\end{center}
\end{table}

\section{Relaxation--based gradient descent--guided search}\label{sec: relaxation}

At each step, the platform needs to first determine the upfront prices for passengers, \emph{i.e.}, the values of the variables in the first-stage problem. After passengers make their decisions, the platform matches passengers who have accepted upfront prices with vehicles. One key property of the stochastic program is that the variables in the first-stage problem are continuous, but those in the second-stage problem are binary. Based on this property, we propose a novel algorithm to solve the two-stage stochastic program named relaxation--based gradient descent--guided search (RGDS). The RGDS algorithm consists of two steps: (1) calculating an initial solution based on relaxation and (2) searching for better solutions guided by gradient descent. In particular, the obtained initial solution from a relaxed optimization problem is feasible for the first-stage problem as its variables are continuous. Also, we can calculate the gradient of the objective function in terms of the continuous variables in the first-stage problem, which can guide the search for better solutions.

Figure \ref{fig: RDGS} illustrates the procedure of the RGDS algorithm. Specifically, we first relax the two-stage stochastic program to a single-stage one fully with continuous variables, and the relaxed problem is denoted as \textbf{RP}. Then, the objective function in the relaxed problem is further approximated using a convex one, denoting it as \textbf{RP-CO}. Consequently, we can obtain an initial solution by solving the RP-CO with low computational costs. Furthermore, we propose a gradient descent--guided search algorithm to enhance the initial solution given by the RP-CO. Specifically, once the solution of the discounts for new passengers is obtained, the probability distribution across all scenarios can be calculated. Moreover, the assignment results in each scenario can be obtained through solving the assignment problem (denoted as \textbf{AP}), \emph{i.e.}, the second-stage problem in the SP, given a specific scenario. As a result, the gradient of the objective function concerning the discount $\bm{\theta}$ can be calculated. Hence, the solution of the discount can be updated using the gradient descent algorithm. However, it is time-consuming to solve the assignment problem involving all scenarios due to (1) the nonconvexity of the AP and (2) the large number of scenarios. Therefore, we relax the assignment variables in the AP to continuous ones. The relaxed AP involving all scenarios (denoted as \textbf{RAP-AS}) is convex and achieves the same objective as the SP. Therefore, it can be used to guide the gradient descent direction. Furthermore, we reduce the scenarios equivalent to the others in terms of the assignment results, significantly accelerating the gradient computation process. Hence, we can obtain the gradient by solving the relaxed assignment problem for a reduced number of scenarios (denoted as \textbf{RAP-RS}).

\begin{figure}[!ht]
    \centering
    \includegraphics[width=0.8\linewidth]{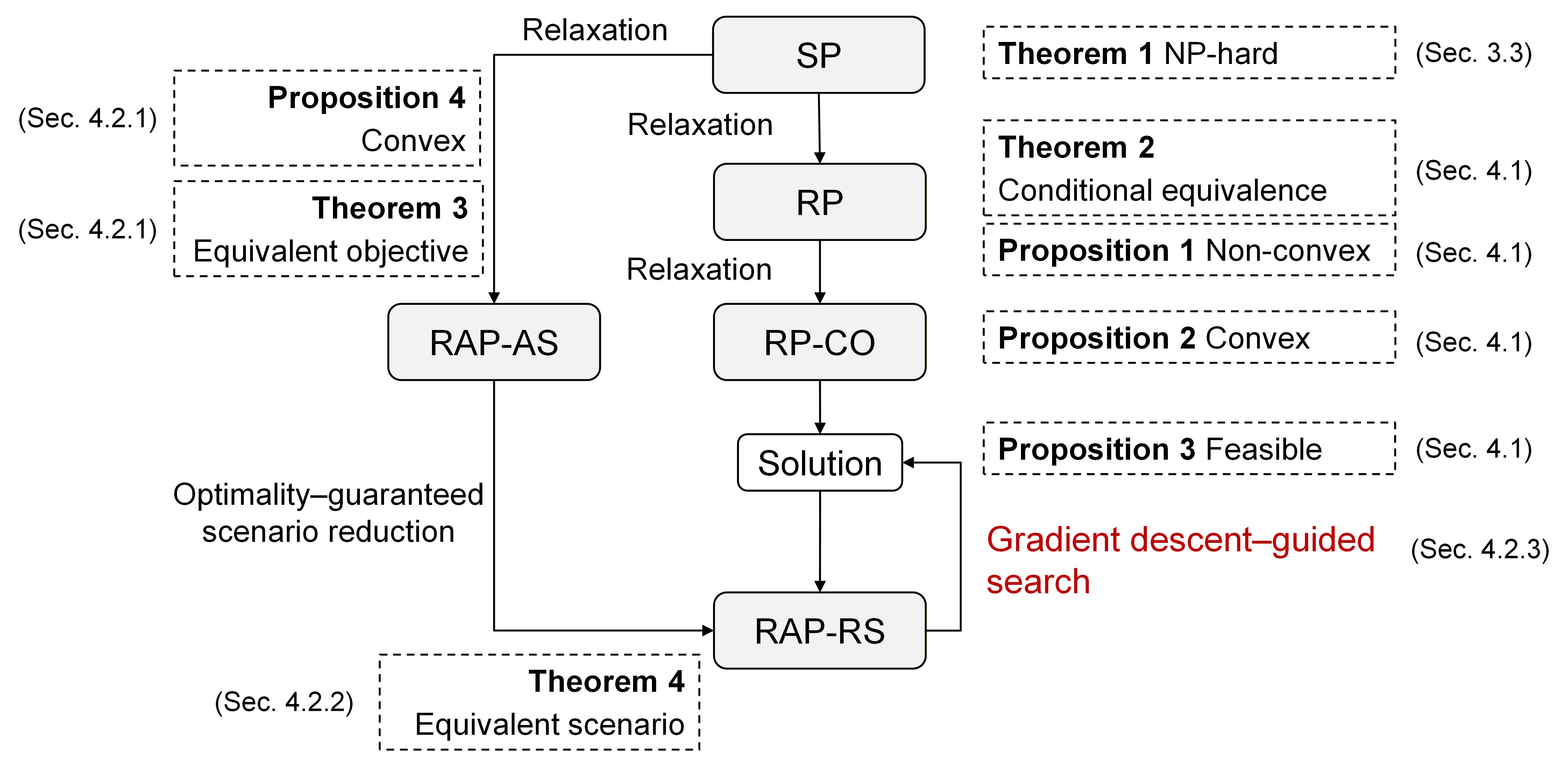}
    \caption{Procedure and main properties of the relaxation--based gradient descent--guided search (RGDS) algorithm.}
    \label{fig: RDGS}
\end{figure}

\subsection{Relaxation--based initial solution generation}

The two-stage stochastic program is first relaxed to a single-stage problem with fully continuous variables. Specifically, the binary variables $y_r$ representing passengers' choices and $x_{kf}$ representing the assignment in the SP are relaxed to continuous ones, \emph{i.e.}, $\Tilde{y}_r \in [0,1], \forall r \in R_n$ and $\Tilde{x}_{kf} \in [0,1], \forall k \in K, f \in F$. In particular, $\Tilde{y}_r$ denotes the probability of a new passenger $r$ who would like to accept the price discount. In other words, $\Tilde{y}_r$ of the passenger is expected to participate in ride-sharing given a discount $\theta_r$, \emph{i.e.}, $\Tilde{y}_r = G(\theta_r)$. Also, the decision variables $\Tilde{x}_{kf}$ determining which trips and vehicles are matched are continuous, as only part of a passenger participating in the platform can be assigned to vehicles. This relaxation enables all scenarios involving passengers' different choices can be considered simultaneously. Consequently, the original two-stage problem can be relaxed to a single-stage one with continuous decision variables, as follows:

\begin{align}
    \min_{\mathbf{\Tilde{x}}, \bm{\theta}} \quad &\sum_{f \in F}\sum_{k\in K}{-\Tilde{x}_{kf} u_{kf}}  \quad &&\textbf{(RP)}    \label{eq: RP} 
\\
\text{s.t.} \quad &
    \sum_{f \in F}{\Tilde{x}_{kf} \leq 1}, \quad &&\forall k\in K, \label{con: RP_C1}
\\
    &\sum_{k\in K}\sum_{r \in f, f \in F}{\Tilde{x}_{kf} \leq G(\theta_r)}, \quad &&\forall r\in R_n, \label{con: RP_C2}
\\
    &\sum_{k\in K}\sum_{r \in f, f \in F}{\Tilde{x}_{kf} \leq 1}, \quad &&\forall r\in R_w,
\\
    &\Tilde{x}_{kf} \in [0, 1], \quad &&\forall f \in F, k\in K,
\\
    &\theta_r \in [\theta^L, \theta^U], \quad &&\forall \theta_r \in \bm{\theta}. \label{con: RP_C5}
\end{align}
Constraints \eqref{con: RP_C2} ensure that the fractional assignment of a passenger does not exceed their individual acceptance probability. It is worth noting that for high-capacity pooled trips containing multiple new passengers, the true joint acceptance probability, assuming independent choices, would be the product of their individual probabilities, \emph{i.e.}, \(\prod_{r \in f} G(\theta_r)\). However, enforcing this joint probability would introduce highly non-linear polynomial terms into the constraints. To maintain computational tractability, RP utilizes the marginal constraints \eqref{con: RP_C2}, which effectively bound the fractional assignment of a pooled trip by \(\min_{r \in f} G(\theta_r)\). While this linear relaxation overestimates the feasibility of shared rides and provides an upper bound on the objective, it serves strictly to generate a high-quality, tractable initial solution. This overestimation is subsequently corrected, and the exact feasibility is refined in the algorithm presented in Section \ref{sec: RGDS}. The other constraints are similar to those in the SP.

As listed in Table \ref{tab: size_problem}, the RP consists of $\mathcal{O}(|K||F|)$ constraints and $\mathcal{O}(|K||F|)$ continuous variables. The RP's size is significantly smaller compared with the SP since it has only one stage. Also, the RP contains only continuous variables, making it easier to solve. The RP uses passengers' acceptance probability as their choices, decoupling the joint probability distribution. This is reasonable when the supply is sufficient, as all passengers can be assigned to vehicles, no matter how many passengers and who will accept ride-sharing, ensuring the overall utility can be maximized across all scenarios.

The relaxed program (RP) can relatively accurately generate the initial solution of the original stochastic program. In particular, as stated in Theorem \ref{thm: equivalence}, when the supply is sufficient, with a few moderate assumptions, the RP and SP are equivalent in terms of upfront discount optimization. It is worth noting that the assumption (c) is reasonable since we have restricted the pickup distance when determining which passengers and vehicles can be potentially matched. In addition, we show that the assumptions (b) and (d) are also reasonable in real-world experiments in Section~\ref{sec: exp}. When the supply is insufficient, we compare the accuracy and computational costs of the RP and SP in Section \ref{sec: compare_problems}, where the RP still achieves comparable accuracy to the SP.

\begin{theorem}\label{thm: equivalence}
Assume
\begin{enumerate}[label=(\alph*)]
    \item the fleet is sufficiently large: $|K|\ge |F|$,
    \item each trip $f\in F$ contains exactly one passenger request,
    \item the pickup costs are negligible, and
    \item the choice probability $G(\theta_r)$ is continuous and strictly increasing in $\theta_r\in[\theta^{L},\theta^{U}]$.
\end{enumerate}
Then the two-stage stochastic program~(SP) and the single-stage convex relaxation \textbf{RP} yield identical optimal discount vectors~$\bm{\theta}^{*}$, \emph{i.e.},
\begin{equation}
\bm{\theta}^{\mathrm{SP}}=\bm{\theta}^{\mathrm{RP}}.
\end{equation}
\end{theorem}

Although the RP is a single-stage optimization problem with continuous decision variables, the computational cost is still high when the number of passengers is large (see Section \ref{sec: compare_problems}), because the objective function in the RP is not convex. Specifically, substituting $u_{kf}$ in Equation \eqref{eq: RP} with Equations \eqref{eq: p_r}--\eqref{eq: utility}, the explicit objective function of RP is as follows:
\begin{equation}
    \label{eq: RP_obj}
    \min_{\Tilde{x}, \bm{\theta}}\sum_{f \in F}\sum_{k\in K}  \sum_{r \in f}{-\left[ (\eta_1 + \eta_2 l_r - \delta \cdot l_{kr})\Tilde{x}_{kf} - (\eta_1 + \eta_2 l_r) \Tilde{x}_{kf}\theta_{r})\right] }.
\end{equation}
It should be noted that the upfront prices for awaiting passengers have been determined in previous steps, and thus, they are not decision variables at the current step. This study unifies the objective function format for notation simplification using the same notation for new and awaiting passengers. The Hessian matrix of the objective function is not positive semi-definite, as it contains only the quadratic term $-\Tilde{x}_{kf}\theta_r$. Consequently, we have the following proposition:

\begin{proposition}\label{prop: RP_not_convex}
    The relaxed program (RP) is non-convex.
\end{proposition}

To further speed up the computation, this study adopts the McCormick envelope \citep{mccormick1976computability} to linearize the quadratic terms in the objective function of the RP, guaranteeing convexity but keeping the bounds sufficiently tight \citep{castro2015tightening}. Specifically, let $w = xy$, and let $x^L$, $x^U$, $y^L$, and $y^U$ denote the lower and upper bound values for $x$ and $y$, respectively. The under-estimators and over-estimators of the function can be represented by $w \geq x^Ly + xy^L - x^Ly^L$, $w \geq x^Uy + xy^U - x^Uy^U$, $w \leq x^Uy + xy^L - x^Uy^L$, and $w \leq xy^U + x^Ly - x^Ly^U$, respectively. Therefore, let $w_{kf}^r = \Tilde{x}_{kf}\theta_r$, then the RP can be further approximated using the McCormick envelope, as follows:
\begin{align}
    \label{eq: RP-CO_obj}
    \min_{\mathbf{\Tilde{x}}, \bm{\theta}} \quad &\sum_{f \in F}\sum_{k\in K}  \sum_{r \in f}{-\left[ (\eta_1 + \eta_2 l_r - \delta \cdot l_{kr})\Tilde{x}_{kf} - (\eta_1 + \eta_2 l_r) w_{kf}^r)\right]} && \textbf{(RP-CO)}
\\
\text{s.t.}
    \label{eq: RP_CO_con1}
    \quad&  w_{kf}^r \geq \theta^L \Tilde{x}_{kf}, &&\forall k \in K, f \in F, r \in f,
\\
    &w_{kf}^r \geq \theta_r + \theta^U \Tilde{x}_{kf} - \theta^U,  &&\forall k \in K, f \in F, r \in f,
\\
    &w_{kf}^r \leq \theta_r + \theta^L \Tilde{x}_{kf} - \theta^L,  &&\forall k \in K, f \in F, r \in f,
\\
    &w_{kf}^r \leq \theta^U \Tilde{x}_{kf},  &&\forall k \in K, f \in F, r \in f,
\\
    \label{eq: RP_CO_con5}
    &\text{Constraints } \eqref{con: RP_C1} \text{--} \eqref{con: RP_C5}.&&
\end{align}

The RP-CO consists of $\mathcal{O}(C|K||F|)$ constraints and $\mathcal{O}(C|K||F|)$ continuous variables (as listed in Table \ref{tab: size_problem}), where $C$ denotes the maximum number of requests that can be scheduled on each trip. Hence, the RP-CO's size is slightly larger than the RP's. However, the objective function in the RP-CO is convex, dramatically reducing its computational costs. Furthermore, we can conclude that the RP-CO is convex, as stated in Proposition \ref{prop: RP-CO_convex}. Since the RP-CO is a single-stage and convex problem, we can obtain the initial solution on short notice by solving the RP-CO.

\begin{proposition}\label{prop: RP-CO_convex}
    The relaxed program (RP-CO) with objective function~\eqref{eq: RP-CO_obj} and constraints~\eqref{eq: RP_CO_con1}--\eqref{eq: RP_CO_con5} is convex.
\end{proposition}

Since the decision variables of the first-stage problem in the SP are continuous, the obtained solution from the relaxed optimization problem is feasible for the first-stage problem, as stated in Proposition \ref{prop: feasible}. This ensures the effectiveness of calculating an initial solution based on relaxation. In addition, numerical studies in Section \ref{sec: compare_problems} demonstrate that the obtained initial solution has a relatively satisfactory accuracy, accelerating the search process.

\begin{proposition}\label{prop: feasible}
    The initial solution obtained from the RP-CO is feasible for the first-stage problem in the SP.
\end{proposition}

\subsection{Solution improvement by gradient descent--guided search}
\label{sec: RGDS}

Once we obtain the solution of the RP-CO, we have the initial solution for the first-stage problem in the SP, \emph{i.e.}, the discounts for new passengers. However, the objective value obtained from the initial solution is an upper bound on the optimal objective, as the assignment variables are relaxed to continuous variables and the joint probabilities of passengers' choices are uncoupled in the RP-CO. Hence, the initial solution can be further enhanced. Since the variables in the first-stage problem in the SP are continuous, we can use gradient descent directions of the objective function in the first-stage problem to guide the search for better upfront price discounts, achieving a lower objective value.

\subsubsection{Deriving the gradient}

To calculate the gradient of the objective function in the first-stage problem at a given solution point, we need to obtain the values of the assignment variables in all scenarios. Specifically, given a solution for the discount, we can determine each scenario $\xi$ ($\forall \xi \in \Xi$) and its corresponding probability $P(\xi)$. Hence, the objective function in the first-stage problem can be formulated as $\sum_{\xi \in \Xi}\sum_{f \in F^\xi}\sum_{k\in K}{-P(\xi) \cdot x_{kf}^\xi \cdot u_{kf}^\xi}$. Moreover, in each scenario, the optimal values of the assignment variables are deterministic and can be obtained by solving the AP, \emph{i.e.}, the second-stage problem in the SP. Consequently, given a solution of the discounts $\bm{\theta}$, the assignment variables $x_{kf}^\xi, \forall \xi \in \Xi$ can be calculated. Therefore, the corresponding gradient of the objective function can also be obtained. However, the computational costs for calculating the gradient can be expensive because the number of scenarios increases exponentially with the number of new passengers, and the AP in each scenario is not convex.

To reduce the computation time, we first relax the integer variables $x_{kf}^{\xi} \in \{0, 1\}$ in the AP to continuous ones: $\hat{x}_{kf}^{\xi} \in [0,1]$. As a result, the relaxed AP involving all scenarios (denoted as RAP-AS) can be modeled as follows: 
\begin{align}
    J = & \min_{\hat{x}^{\xi}}\sum_{\xi \in \Xi}\sum_{f \in F^\xi}\sum_{k\in K}{-P(\xi) \cdot \hat{x}_{kf}^\xi \cdot u_{kf}^\xi} && \textbf{(RAP-AS)}
    \\
\text{s.t.} \quad
    & \hat{x}_{kf}^\xi \in [0,1],  &&\forall f \in F^\xi, k\in K,
\\
    & \text{Constraints } \eqref{con: SP_C1} \text{--} \eqref{con: SP_C3}. &&
\end{align}
The probability $P(\xi)$, utility $u_{kf}^{\xi}$, and passengers' choices $y_r^{\xi}$ are deterministic in each scenario once the discounts for passengers are given, leading to a linear objective function as well as linear constraints. Hence, the relaxation ensures the convexity of RAP-AS, as stated in Proposition \ref{prop: RSP_AS_convex}. In addition, as stated in Theorem \ref{thm: J = O}, the obtained optimal objective value from the RAP-AS is equivalent to that from the first-stage problem in the SP when $\bm{\theta}$ is given. Therefore, we can search for better solutions guided by the gradient descent direction of the objective function in RAP-AS. We further explain this theorem using a toy example in~\ref{seca: toy_example}. Furthermore, as listed in Table \ref{tab: size_problem}, RAP-AS consists of $\mathcal{O}(|\Xi||K||F|)$ constraints and $\mathcal{O}(|\Xi||K||F|)$ continuous variables. Considering the number of scenarios $|\Xi|$ increases exponentially with the number of new passengers, the scenarios should be reduced to further improve computation speed.

\begin{proposition}\label{prop: RSP_AS_convex}
    The relaxed assignment problem involving all scenarios (RAP-AS) is convex.
\end{proposition}

\begin{theorem}\label{thm: J = O}
    Fix the discounts for passengers $\bm{\theta}$, the RAP-AS has an optimal integer solution, which achieves the same optimal objective as the SP, \emph{i.e.}, $J = O$.
\end{theorem}

The gradient can be calculated as follows:
\begin{equation}\label{eq: J_gradient}
    \nabla_{\bm{\theta}} J = \frac{\partial J}{\partial \bm{\theta}}
    = \sum_{\xi \in \Xi}\sum_{f \in F^\xi}\sum_{k\in K}-x_{kf}^{\xi}
    \left[
    \frac{dP(\xi)}{d\bm{\theta}}u_{kf}^{\xi} + P(\xi) \frac{du_{kf}^{\xi}}{d\bm{\theta}}
    \right].
\end{equation}
The $r$-th component is
\begin{equation}
        \frac{\partial J}{\partial \theta_r}
            = \sum_{\xi \in \Xi}\sum_{f \in F^\xi}\sum_{k\in K}
            -x_{kf}^{\xi}
        \left[
        \frac{\partial P(\xi)}{\partial \theta_r}u_{kf}^{\xi}
        + P(\xi) \frac{\partial u_{kf}^{\xi}}{\partial \theta_r}
        \right].
\end{equation}
We now calculate the derivative of $P(\xi)$ and $u_{kf}^\xi$ w.r.t. $\theta_r$, respectively. Specifically, for fixed $r \in R_n$,
\begin{equation}
    \log P(\xi;\boldsymbol{\theta})
= \sum_{r' \in R_n}
  \Bigl(
    y_{r'}^\xi \log G(\theta_{r'})
    + (1-y_{r'}^\xi)\log(1-G(\theta_{r'}))
  \Bigr).
\end{equation}
Hence
\begin{equation}
    \frac{\partial \log P(\xi;\boldsymbol{\theta})}{\partial \theta_r}
= y_r^\xi \frac{G'(\theta_r)}{G(\theta_r)}
  - (1-y_r^\xi)\frac{G'(\theta_r)}{1-G(\theta_r)},
\end{equation}
and therefore, we have:
\begin{equation}\label{eq:dP_dtheta}
\frac{\partial P(\xi;\boldsymbol{\theta})}{\partial \theta_r}
= P(\xi;\boldsymbol{\theta}) \, G'(\theta_r)
  \left(
    \frac{y_r^\xi}{G(\theta_r)}
    - \frac{1-y_r^\xi}{1-G(\theta_r)}
  \right).
\end{equation}
Also, we have
\begin{equation}
    u_{kf}^\xi(\boldsymbol{\theta})
= \sum_{r' \in f}\bigl(p_{r'}(\theta_{r'}) - \delta l_{kr'}\bigr)
= \sum_{r' \in f}\bigl((1-\theta_{r'})(\eta_1 + \eta_2 l_{r'}) - \delta l_{kr'}\bigr).
\end{equation}
Consequently, $u_{kf}^\xi$ depends on $\theta_r$ only if $r\in f$. Therefore, we have:
\begin{equation}\label{eq:du_dtheta}
\frac{\partial u_{kf}^\xi}{\partial \theta_r}
=
\begin{cases}
-(\eta_1 + \eta_2 l_r), & \text{if } r \in f,\\[2pt]
0,                      & \text{if } r \notin f.
\end{cases}
\end{equation}
Substituting Equations~\eqref{eq:dP_dtheta} and \eqref{eq:du_dtheta}, we obtain
\begin{align}
\frac{\partial J}{\partial \theta_r}
& = -\sum_{\xi \in \Xi}\sum_{f \in F^\xi}\sum_{k\in K}
   x_{kf}^{\xi}
   \left[
      u_{kf}^{\xi} \, P(\xi) \, G'(\theta_r)
      \left(
        \frac{y_r^\xi}{G(\theta_r)} - \frac{1-y_r^\xi}{1-G(\theta_r)}
      \right)
      + P(\xi) \,\mathbf{1}_{\{r\in f\}}\, \bigl(-(\eta_1 + \eta_2 l_r)\bigr)
   \right],
\end{align}
where $\mathbf{1}_{\{r\in f\}}$ is the indicator that request $r$ belongs to
trip $f$. Equivalently, grouping $P(\xi)$, we have
\begin{equation}\label{eq: gradient}
    \frac{\partial J}{\partial \theta_r} = -\sum_{\xi \in \Xi} P(\xi)
    \bigg[
      G'(\theta_r)
      \left(
        \frac{y_r^\xi}{G(\theta_r)}
        - \frac{1-y_r^\xi}{1-G(\theta_r)}
      \right)
      \sum_{f \in F^\xi}\sum_{k\in K} x_{kf}^{\xi} u_{kf}^{\xi}
      - (\eta_1 + \eta_2 l_r)
        \sum_{f \in F^\xi :\, r\in f}\sum_{k\in K} x_{kf}^{\xi}
    \bigg].
\end{equation}
Given $\bm{\theta}$ and assignment results $x^\xi$ in each scenario, we can calculate the gradient of the objective using Equation~\eqref{eq: gradient}, which can guide the search direction to obtain better results.

\subsubsection{Optimality--guaranteed scenario reduction}

In this section, we identify scenarios where the values of the assignment variables in the second stage of SP can be directly obtained from other scenarios. As a result, we do not need to solve the AP in these scenarios and can obtain the optimal solutions. This leads to lower computational costs without sacrificing accuracy. In particular, when the supply is insufficient, \emph{i.e.}, $|K| < |F|$, at most $|K|$ trips can be accommodated since each vehicle can match with at most one trip. In this context, the assignment results of scenarios involving more than $|K|$ trips can be obtained from the subproblems involving no more than $|K|$ trips. Specifically, we solve the AP to obtain the values of the assignment variables involving the vehicles $K$ and $n$ trips in $F$, where $n = 1, 2, ..., |K|$. These results cover the assignment results of the scenarios where $|F^{\xi}| > |K|$ ($\forall \xi \in \Xi$), as at most $|K|$ trips in each scenario can be assigned to vehicles. For example, as shown in Figure \ref{fig: operation} (b), the platform needs to match two vehicles with four trips. In this context, the optimal assignment result (as shown in Figure \ref{fig: operation} (c)) can be obtained from scenarios involving two or fewer trips, since at most two trips can be matched with vehicles. To be more specific, we have the following theorem:

\begin{theorem}\label{thm: scenario_reduction}
    Let $\xi$ be a scenario with a vehicle set $K$ and a trip set $F^{\xi}$ such that $|F^{\xi}| > |K|$.  There exists a subset $F' \subset F^{\xi}$ with $|F'| \le |K|$ such that the optimal assignment between $F^{\xi}$ and $K$ coincides with the optimal assignment between $F'$ and $K$; in particular, all trips in $F^{\xi}\setminus F'$ are optimally left unassigned.
\end{theorem}

The above theorem enables us to solve the RAP-AS with reduced scenarios (denoted as RAP-RS), which can reduce the computational costs for calculating the gradient. In particular, the reduced scenarios can be represented as $\Xi^{\text{RS}} = \{F' | F' \subset F^{\text{all}}, |F'| \leq |K|\}$. Hence, as listed in Table \ref{tab: size_problem}, RAP-RS consists of $\mathcal{O}(|\Xi^{\text{RS}}||K||F|)$ constraints and $\mathcal{O}(|\Xi^{\text{RS}}||K||F|)$ continuous variables, leading to a smaller problem than the RAP-AS. We use the RAP-RS to calculate the direction of the gradient descent and update the solution along the gradient descent direction.

\subsubsection{Gradient descent--guided search}

Starting from an initial vector of upfront discounts, we iteratively update the solution along the negative gradient direction of the RAP-AS objective function. Let $\gamma$, $\epsilon$, and $S$ denote the learning rate, the minimum relative improvement required in each iteration, and the maximum number of search steps, respectively. In each iteration, we (i) compute the gradient of the current objective value with respect to $\bm{\theta}$, (ii) take a gradient descent step, (iii) re-evaluate the scenario probabilities and solve RAP-RS to obtain the new objective value. The search procedure terminates if the relative reduction of the RAP-AS objective is smaller than $\epsilon$, or if the iteration counter reaches $S$. The complete procedure is summarized in Algorithm~\ref{alg: RGDS}.

\begin{algorithm}[!ht]
    \caption{RGDS}
    \label{alg: RGDS}
    \SetAlgoLined
    \KwIn{Vehicle set $K$; trip set$F$; assignment utilities $u$; learning  rate $\gamma > 0$; minimum improvement ratio $\epsilon > 0$; maximum number of search iterations $S \in \mathbb{N}$.}
    \KwOut{Upfront discount $\bm{\theta}$ for new passengers.}
    
    Obtain initial solution $\bm{\theta^0}$ by solving the RP-CO \;
    Calculate the scenario probability distribution using Equation \eqref{eq: P_xi}\;
    Obtain the assignment results and objective $J^0$ by solving the RAP-RS \;
    $s \gets 0$ \;
    
    \While{$s < S$}{
        Calculate the objective function's gradient 
        $\nabla_{\bm{\theta}^s} J^s$ using Equation~\eqref{eq: gradient}\;
    
        $\bm{\theta}^{s+1}
        \gets
        \operatorname{CLIP}
        \left(
        \bm{\theta}^s - \gamma \nabla_{\bm{\theta}^s} J^s,
        \bm{\theta}^L,
        \bm{\theta}^U
        \right)$\;
    
        Calculate the scenario probability distribution using Equation~\eqref{eq: P_xi}\;
    
        Obtain the assignment results and objective $J^{s+1}$ by solving the RAP-RS\;
    
        \If{$J^{s+1} < (1+\epsilon)J^s$}{
            $s \gets s+1$\;
        }
        \Else{
            \Return $\bm{\theta}^s$\;
        }
    }
    \Return $\bm{\theta^s}$ \;
\end{algorithm}

\section{Numerical study among different approaches}
\label{sec: compare_problems}

In this section, we demonstrate the effectiveness and efficiency of the proposed algorithm using several case studies. Specifically, we first determine the most suitable learning rate for the proposed RGDS algorithm in Section~\ref{sec: learning_rate}. In addition, we compare the proposed RGDS algorithm with the commonly used local search algorithm to validate its efficiency in Section~\ref{sec: compared_with_RINS}. Furthermore, we compare the objective gap and the computation time between the proposed algorithm and the original stochastic program in Section~\ref{sec: numerical_results}.

The trip fares of passengers and pickup costs are generated randomly in all cases. Specifically, the original prices for new passengers (without discount) range from 20 to 100, the discounted prices for previous passengers range from 10 to 20, and the pickup costs range from 0 to 5. To simplify the numerical experiments, we assume that each trip contains only one passenger, and all drivers can be matched with all passengers. We consider multiple passengers sharing one trip and pickup distance constraints in real-world simulations in Section~\ref{sec: exp}. We conduct all experiments on a workstation with Intel(R) Xeon(R) Silver 4314 CPU @ 2.40GHz.

\subsection{Setting the learning rate in RGDS}
\label{sec: learning_rate}

We examine the effect of the learning rate on the convergence behavior of the proposed gradient-based solution refinement procedure. In all experiments, an initial solution is obtained using the RP-CO, which is subsequently improved through iterative updates guided by gradient-descent directions with different learning rates. Two representative scenarios are considered: (i) an instance with 8 passengers and 3 drivers, and (ii) an instance with 10 passengers and 4 drivers.

\begin{figure}[!ht]
    \centering
    \includegraphics[width=0.8\linewidth]{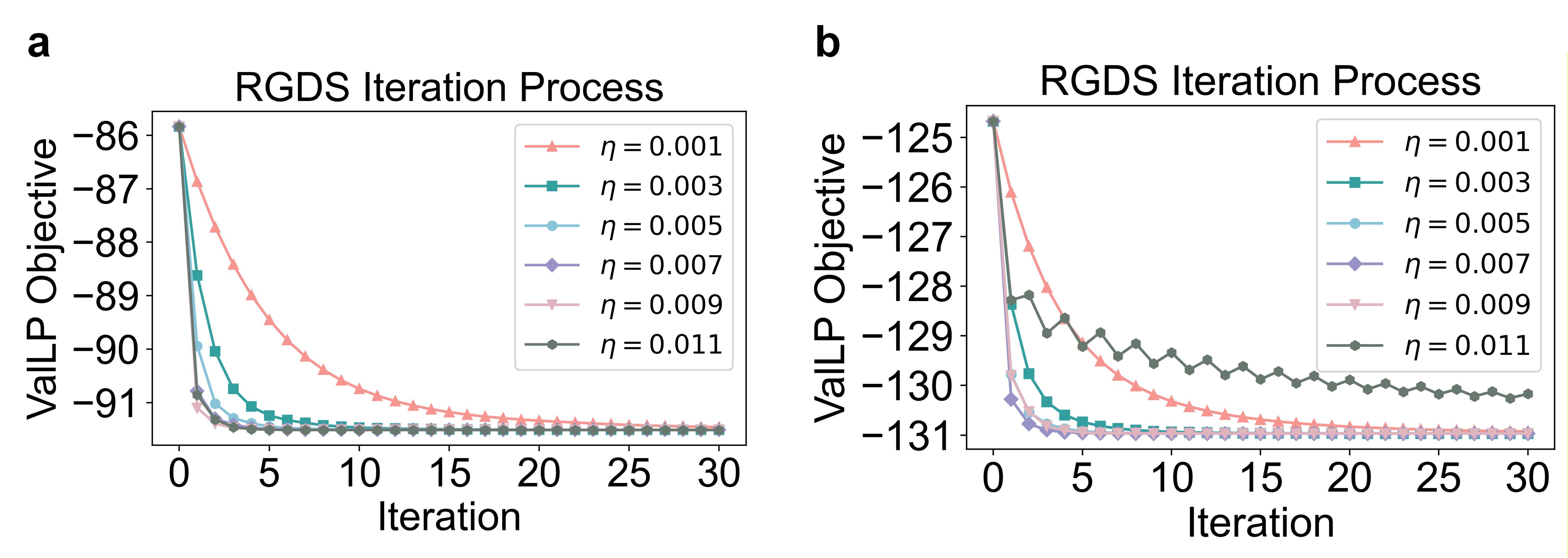}
    \caption{Iteration process with different learning rates involving (a) 8 passengers and 3 drivers and (b) 10 passengers and 4 drivers.}
    \label{fig: RDGS_lr}
\end{figure}

Figure~\ref{fig: RDGS_lr} reports the evolution of the objective value over iterations for learning rates $\gamma \in \{0.001, 0.003, 0.005, 0.007, 0.009, 0.011\}$. For the case with 8 passengers and 3 drivers, learning rates $\gamma = 0.007$, 0.009, and 0.011 exhibit similar convergence trends and achieve substantially faster objective reduction than smaller learning rates. A similar pattern is observed in Figure~\ref{fig: RDGS_lr}(b) for the scenario with 10 passengers and 4 drivers. While larger learning rates accelerate early-stage improvement, excessively large values (\emph{e.g.}, $\gamma = 0.011$) lead to oscillatory behavior and limit further objective reduction. On the contrary, $\gamma = 0.007$ provides a favorable balance between convergence speed and solution quality across both scenarios. Based on these results, we select $\gamma = 0.007$ as the learning rate for all subsequent experiments.




\subsection{Convergence performance: comparing RGDS with the RINS baseline}
\label{sec: compared_with_RINS}

In addition to the proposed RGDS method, we consider a Relaxation-Induced Neighborhood Search (RINS) approach as a benchmark~\citep{danna2005exploring}. RINS starts from the same initial solution obtained by solving RP-CO and iteratively improves the solution by locally perturbing one decision variable at a time. Unlike RGDS, which updates all decision variables simultaneously along the gradient direction, RINS performs coordinate-wise searches with a fixed step size and accepts only improving moves. We formulate the RINS algorithm in~\ref{seca: RINS}. We compare RGDS and RINS under the same experimental setting as depicted in Section~\ref{sec: learning_rate}. For both methods, the initial solution is identical, and the maximum number of search iterations is set to 20. For RINS, we evaluate multiple learning rates (step sizes) to examine the trade-off between convergence stability and solution quality.

\begin{figure}[!ht]
    \centering
    \includegraphics[width=0.8\linewidth]{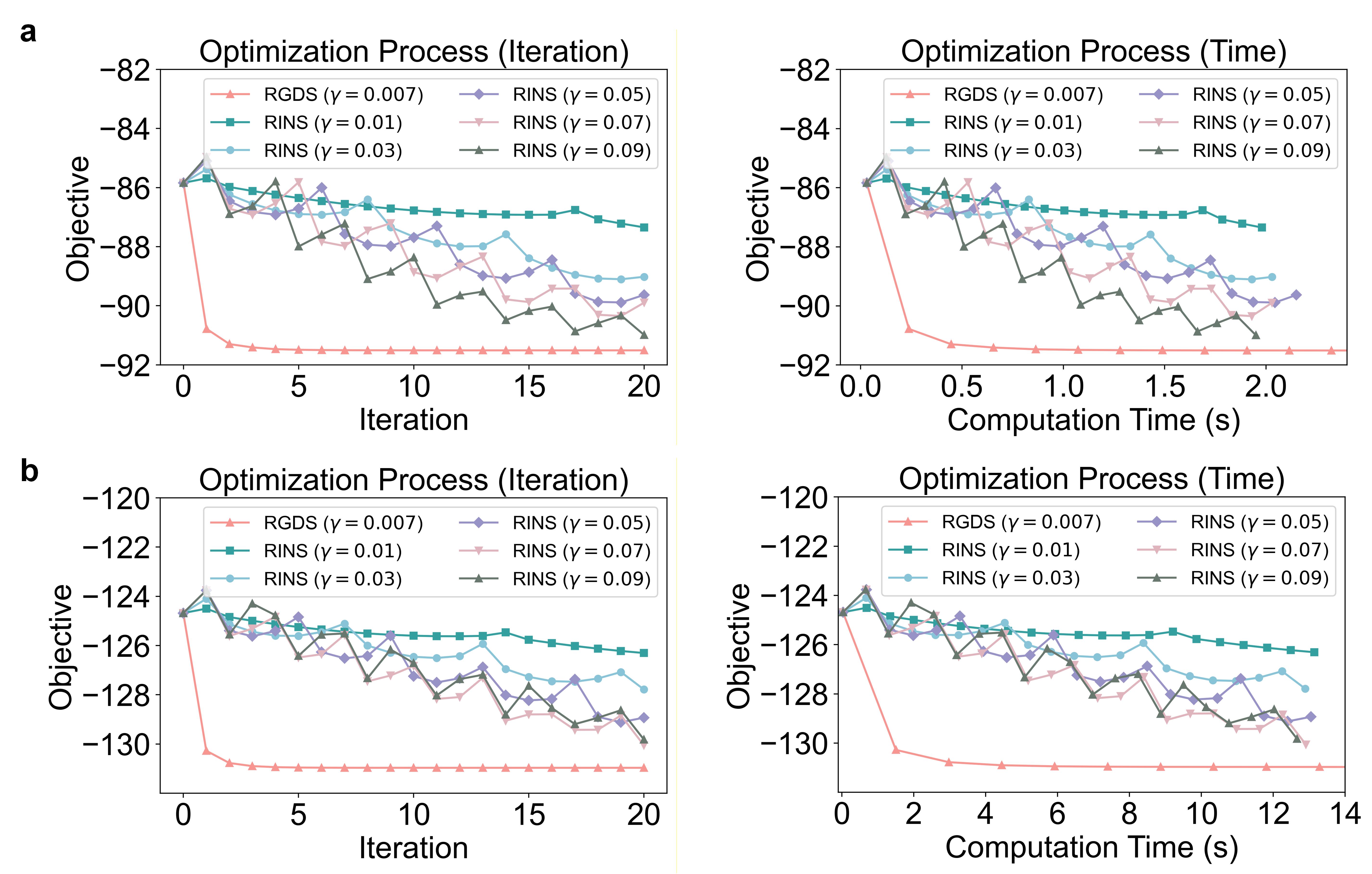}
    \caption{Comparison of RGDS and RINS under different learning rates in terms of iteration-based and time-based convergence involving (a) 8 passengers and 3 drivers and (b) 10 passengers and 4 drivers.}
    \label{fig: RDGS_RINS}
\end{figure}

Figure~\ref{fig: RDGS_RINS} compares the convergence behavior between RGDS and RINS under a broader range of learning rates. When the learning rate is low, RINS exhibits stable but slow convergence and achieves limited improvement over the initial solution. As the learning rate increases, RINS attains lower objective values; however, the convergence behavior becomes increasingly unstable. Notably, even with larger learning rates, the best objective values obtained by RINS remain inferior to those achieved by RGDS. In addition, RINS with $\gamma = 0.07$ achieves satisfactory convergence performance in both scenarios. Hence, we use 0.07 as the learning rate for RINS in the subsequent experiments. In terms of computation time, owing to its relatively inexpensive coordinate-wise updates, RINS is able to perform more iteration steps within the same time budget. However, despite this higher iteration frequency, the objective value decreases only marginally over time, indicating slow convergence in terms of solution quality. On the contrary, RGDS requires fewer iterations but achieves substantially faster objective reduction per unit of time. By leveraging gradient-guided simultaneous updates, RGDS efficiently exploits descent directions and rapidly converges to a superior solution. As a result, RGDS consistently attains lower objective values than RINS within the same computational time horizon. The comparison highlights the advantage of RGDS in terms of convergence speed, computational time efficiency and solution quality, validating the effectiveness of gradient-based global updates over local coordinate-wise search.

\subsection{Effectiveness of the proposed RGDS algorithm: runtime performance and solution quality}
\label{sec: numerical_results}

To validate the effectiveness of the proposed algorithms, we conduct computational experiments with a few awaiting passengers, new passengers, and vehicles, denoted by $|R_w|$, $|R_n|$, and $|K|$, respectively. To facilitate the experiments, this study restricts the maximum computation time for the two-stage optimization program (SP) and the relaxed program (RP) to four hours (\emph{i.e.}, 14,400 seconds). For the two search algorithms, RGDS and RINS, we impose similar iteration budgets to balance solution quality and computational effort: the maximum number of iterations is set to 5 for RGDS and 10 for RINS, resulting in comparable computational time. The termination threshold for improvement is fixed at 0.2\%. For RGDS, we further report results with both full and reduced scenarios; as guaranteed by the scenario reduction theorem, both variants follow the same iteration process, differing only in computation time.

\begingroup
\setlength{\tabcolsep}{6pt} 
\renewcommand{\arraystretch}{1} 
\begin{table}[!ht]
\caption{Comparison of solution quality and computational efficiency among SP, RP, RP-CO, RINS, and RGDS under different problem scales. We report only one achieved objective (\emph{i.e.}, $O_5$) of RGDS with full or reduced scenarios, as their iteration processes are the same except for the computation time (reported as $T_5$ and $T_5^{\text{RS}}$, respectively). A negative ``Gap'' means the proposed RGDS or RINS algorithm outperforms the original SP with the computation time constraint.}
\begin{center}
\resizebox{1.\columnwidth}{!}{
\begin{tabular}{ l | r r | r r | r r | r r r | r r r r}
\toprule
 & \multicolumn{2}{c |}{SP} & \multicolumn{2}{c |}{RP} & \multicolumn{2}{c |}{RP-CO} & \multicolumn{3}{c |}{RINS} & \multicolumn{4}{c}{RGDS}\\
\midrule
$|R_w|, |R_n|, |K|$  & $T_1$ (s) & $O_1$ & $T_2$ (s) & $O_2$ & $T_3$ (s) & $O_3$ & $T_4$ (s) & $O_4$ & Gap (\%) & $T_5$ (s) & $T_5^{\text{RS}}$ (s) & $O_5$ & Gap (\%) \\
\midrule

0, 1, 1 & 0.01 & -19.27 & 0.01 & -19.27 & 0.01 & -19.17 & 0.01 & -19.17 & 0.50 & 0.02 & 0.01 & -19.27 & 0.00 \\
0, 2, 1 & 0.05 & -18.10 & 0.04 & -18.08 & 0.01 & -18.04 & 0.01 & -18.04 & 0.31 & 0.01 & 0.01 & -18.10 & 0.02 \\
0, 3, 1 & 3.69 & -36.63 & 0.03 & -36.06 & 0.01 & -36.05 & 0.02 & -36.57 & 0.16 & 0.02 & 0.02 & -36.57 & 0.15 \\
0, 4, 1 & 14400.00 & -22.51 & 0.05 & -21.77 & 0.01 & -22.10 & 0.04 & -22.46 & 0.21 & 0.05 & 0.02 & -22.46 & 0.22 \\
1, 5, 2 & 14400.00 & -94.70 & 0.19 & -92.53 & 0.02 & -91.24 & 0.17 & -94.40 & 0.32 & 0.54 & 0.21 & -94.58 & 0.13 \\
1, 6, 2 & 14400.00 & -78.37 & 0.61 & -77.80 & 0.02 & -74.05 & 0.26 & -76.28 & 2.67 & 0.22 & 0.13 & -77.53 & 1.08 \\
1, 7, 3 & 14400.00 & -103.32 & 39.11 & -102.80 & 0.07 & -100.01 & 0.97 & -102.41 & 0.88 & 0.83 & 0.45 & -103.03 & 0.28 \\

\midrule
Average & \multicolumn{7}{r}{0.21}  & \multicolumn{2}{r}{0.72} & 0.24 & 0.12 & & 0.27\\
\midrule

1, 8, 3 & 14400.00 & -128.07 & 266.41 & -126.85 & 0.03 & -120.34 & 1.38 & -125.17 & 2.26 & 1.88 & 1.05 & -128.47 & -0.32 \\
1, 9, 4 & 14400.00 & -133.11 & 14400.00 & -143.03 & 0.06 & -139.59 & 4.68 & -141.99 & -6.67 & 4.01 & 2.68 & -142.97 & -7.41 \\

\bottomrule

\end{tabular}}
\label{tab: relax_val}

\end{center}
\end{table}

Table~\ref{tab: relax_val} compares SP, RP, RP-CO, RINS, and RGDS in terms of solution quality and computational efficiency. Consistent with previous observations, SP can only solve very small instances within the time limit. For slightly larger instances, SP either reaches the four-hour limit or returns suboptimal solutions, highlighting its limited scalability. Similarly, RP struggles with medium- and large-scale instances due to the non-convexity of its objective function, and in some cases also fails to terminate within the time limit. On the contrary, RP-CO solves all tested instances within a fraction of a second, confirming that the convex relaxation dramatically improves tractability while providing reasonably good objective values.

Starting from RP-CO, both RINS and RGDS further improve solution quality through iterative search. On average, RGDS achieves smaller optimality gaps relative to SP than RINS, while maintaining comparable or lower computation times. As summarized in the average results of Table~\ref{tab: relax_val}, RGDS consistently outperforms RINS in terms of solution quality, yielding lower objective gaps and demonstrating more stable performance across instances. This indicates that incorporating global gradient information enables RGDS to explore the solution space more effectively than the local, coordinate-wise search adopted by RINS.

For large-scale instances where SP cannot converge within the four-hour time limit, both RINS and RGDS are able to produce high-quality solutions efficiently. Notably, RGDS frequently achieves better objective values than SP under the imposed time constraint, as reflected by negative optimality gaps in the largest instances. This result validates the effectiveness of RGDS in identifying high-quality solutions for large-scale problems that are otherwise intractable for exact formulations. Moreover, the reduced-scenario version of RGDS substantially lowers computation time while preserving the same solution trajectory, further confirming the scalability and practical applicability of the proposed approach.

Overall, these results demonstrate that while RINS provides meaningful improvements over convex relaxation by exploiting local search, RGDS consistently delivers superior solution quality and robustness by leveraging global gradient information, making it more suitable for large-scale dynamic ride-sharing optimization.

\section{Large-scale real-world simulation experiments}
\label{sec: exp}

In this section, we conduct experiments in Chengdu and Shanghai, two metropolises in China, to demonstrate the effectiveness of the proposed joint optimization framework. Specifically, we introduce the experiment setup, including the data, various scenarios, and parameters, in Section~\ref{sec: exp_setup}. The overall experimental results are introduced in Section~\ref{sec: overall_results}. Subsequently, the service rate and spatiotemporal discounts are analyzed in Sections~\ref{sec: service_rate} and~\ref{sec: spatiotemporal_analysis}, respectively. Finally, the pooling sizes and distances are discussed in Section~\ref{sec: pooling_size}.

\subsection{Experiment setup}
\label{sec: exp_setup}

The transportation demand data used in this study are provided by a large Chinese transportation network company (TNC) and contain ride requests collected in two major metropolitan areas: Chengdu and Shanghai. Each request records the geographic coordinates and timestamps of its origin and destination, which allows us to reconstruct both the spatial and temporal characteristics of travel demand. The underlying road networks of the two cities are obtained from OpenStreetMap \citep{OpenStreetMap}. As shown in Figure~\ref{fig: demand_spatial}, both networks exhibit dense and highly connected topologies in the central urban areas, with a large number of intersections and arterial roads forming fine-grained grids. Such network structures are representative of large metropolitan cities and are well suited for evaluating dynamic ride-sharing operations under realistic traffic conditions. Considering that running large-scale simulations on the full dataset is computationally expensive, we sample 20\% of requests from the original data and focus on the central areas of the two cities during a representative peak period, \emph{i.e.}, from 4 to 8 pm. Consequently, the overall numbers of requests are 12,737 and 14,295, and the average trip distances are 7.5 and 6.5 km in Chengdu and Shanghai, respectively. In addition, the study areas are 330 and 432 km$^2$, respectively.

\begin{figure}[!ht]
    \centering
    \includegraphics[width = 0.9\textwidth]{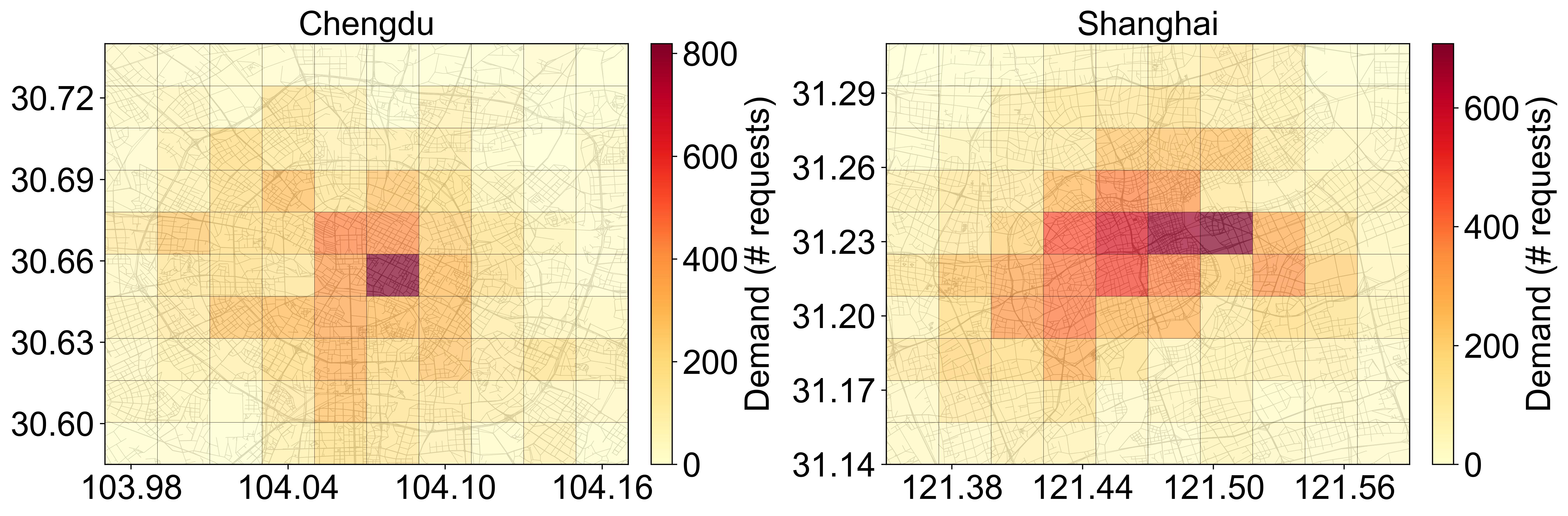}
    \caption{Road networks and spatial distributions of requests in Chengdu and Shanghai.}
    \label{fig: demand_spatial}
\end{figure}

Figure~\ref{fig: demand_spatial} also illustrates the spatial distribution of ride requests in Chengdu and Shanghai. In both cities, demand is highly concentrated in the central urban regions, gradually decreasing toward the periphery. This spatial heterogeneity reflects typical commuting and activity patterns in large cities and highlights the importance of effectively matching supply and demand in dense areas. Notably, Shanghai exhibits a more elongated high-demand region along major urban corridors, while Chengdu shows a relatively compact demand core. Figure~\ref{fig: demand_arrival_rate} illustrates the temporal evolution of request arrival rates in Chengdu and Shanghai between 16:00 and 20:00. In both cities, demand increases noticeably after 17:00, indicating the onset of the evening peak period. While Shanghai experiences sustained high demand from approximately 17:00 until 20:00, Chengdu exhibits a shorter peak: its arrival rate reaches a maximum around 18:00 and then declines after approximately 18:30. This contrast highlights differences in peak-hour demand duration between the two cities. Overall, the observed temporal patterns are representative of typical evening commuting behavior and provide a suitable basis for evaluating ride-sharing performance under dynamic and time-varying demand conditions.

\begin{figure}[!ht]
    \centering
    \includegraphics[width = 0.5\textwidth]{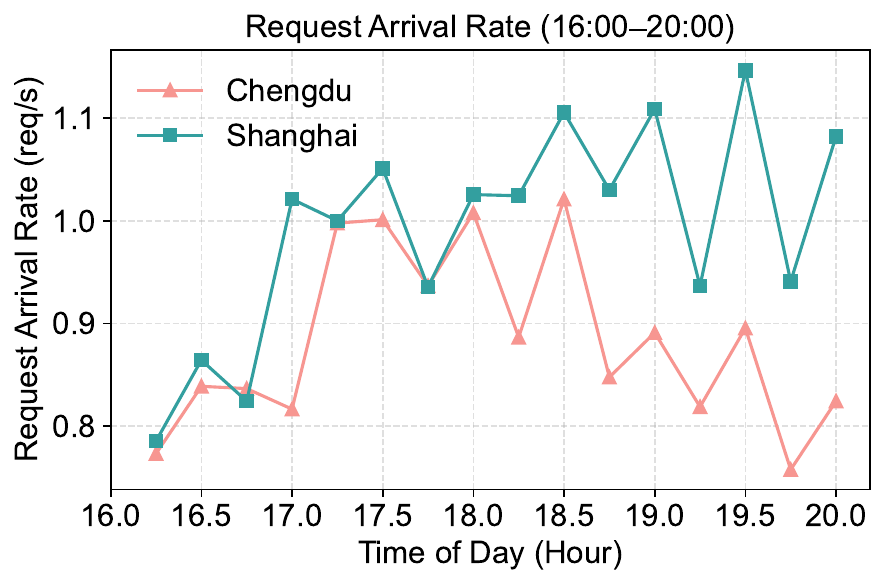}
    \caption{Temporal distributions of request arrival rates in Chengdu and Shanghai.}
    \label{fig: demand_arrival_rate}
\end{figure}

We consider two fleet size scenarios with 100 and 200 vehicles, representing markets with insufficient and relatively sufficient supply, respectively. In addition, two vehicle capacity settings are examined: a low-capacity scenario with a maximum of two passengers per vehicle ($C=2$) and a high-capacity scenario with four passengers per vehicle ($C=4$). The basic fare is set to \$5 per trip, and the distance-based fare is \$3 per kilometer. The vehicles are initialized according to the distributions of requests. Vehicles are assumed to travel along shortest paths on these networks at a constant speed of 8~m/s, and the shortest paths are computed using Dijkstra’s algorithm \citep{dijkstra1959note}. All experiments are conducted on HRSim \citep{chen2025hrsim}, a high-capacity ride-sharing simulator that can achieve large-scale simulations.

As the benchmark, we adopt an operation strategy with a fixed discount, originally proposed by \citet{alonso2017demand} and widely used in subsequent studies \citep{shah2020neural, yang2024prediction}. Under this benchmark, a uniform upfront discount is offered to all passengers. We evaluate a series of fixed discount levels ranging from 0.2 to 0.5 and report the best performance achieved by the benchmark for comparison with the proposed optimization framework.
In addition, we also use the heuristic algorithm introduced before, RINS, to optimize the discounts, which is considered the second benchmark.
We evaluate system performance using two key metrics: the service rate (\textbf{SR}) and the average revenue per vehicle (\textbf{AR}). The service rate is defined as the ratio of the number of satisfied passengers to the total number of passengers, including those who decline the upfront prices. The average revenue is computed as the mean revenue earned by each vehicle over the simulation horizon.

\subsection{Overall results}
\label{sec: overall_results}

Figure \ref{fig: perf} demonstrates the system performance of the benchmark with various fixed discounts and the proposed joint optimization method. The proposed method achieves a higher average revenue than the benchmark across all scenarios. Also, in each scenario where the benchmark achieves the maximum revenue with a fixed discount, the proposed method can satisfy more passengers than the benchmark (see Table \ref{tab: comp_perf} for detailed comparison). These results demonstrate that the proposed method can improve the traffic efficiency of dynamic ride-sharing in various scenarios involving different fleet sizes, vehicle capacities, and demand patterns, increasing the overall revenue and service rates. It should be noted that the service rate (SR) is relatively low as it is the ratio of the answered requests to the total requests in the market. This is reasonable because the market share of ride-sharing services is relatively low since passengers may decline ride-sharing services \citep{chen2024development}.

\begin{figure}[!ht]
    \centering
    \includegraphics[width = 0.85\textwidth]{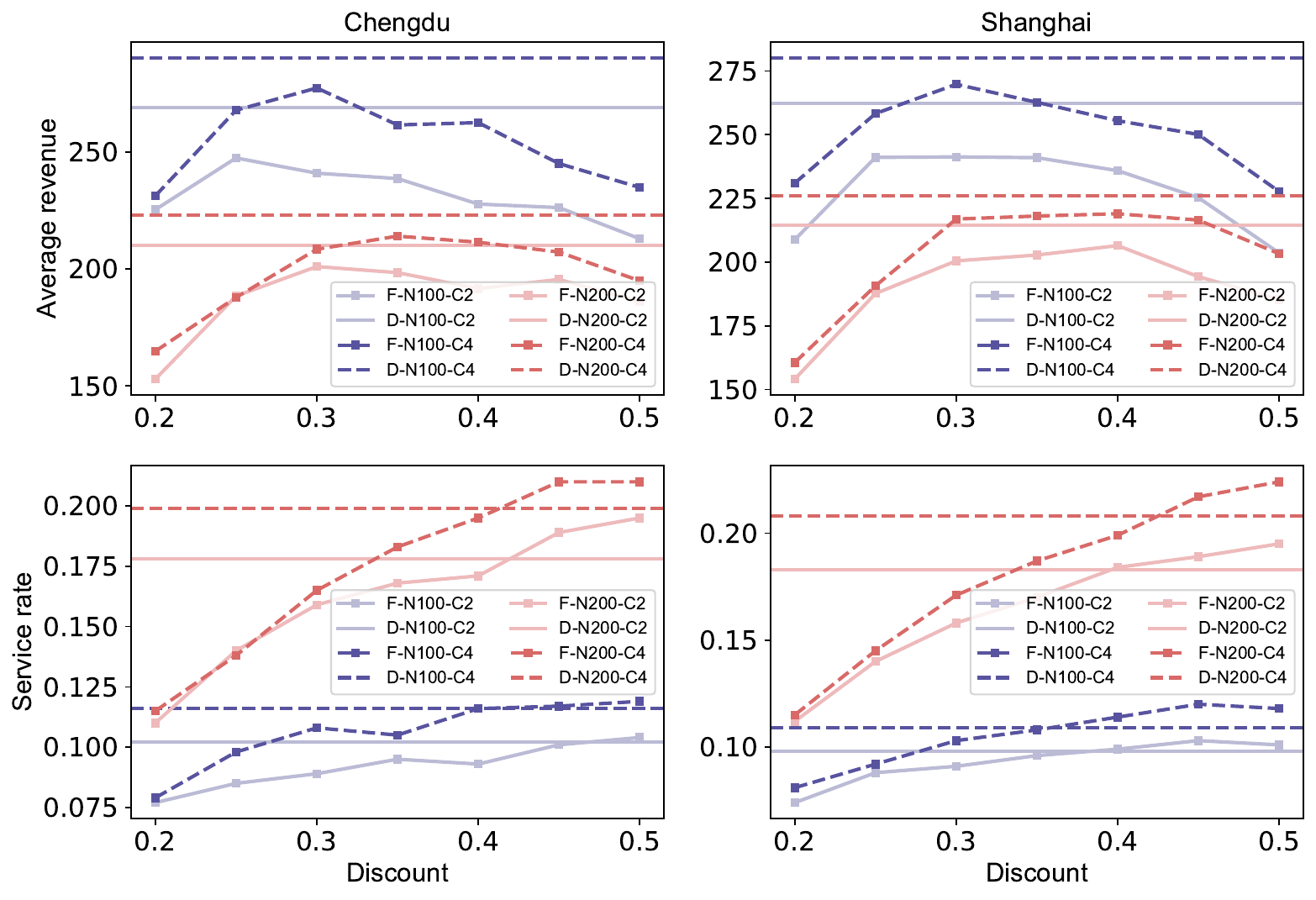}
    \caption{Comparison between the proposed optimization method (denoted as D) and the benchmark (denoted as F) involving various discounts on the average revenue of each vehicle and average service rate. N and C denote the number of vehicles and vehicle capacity (similarly hereinafter), respectively.}
    \label{fig: perf}
\end{figure}

Regarding the fixed pricing strategy, on the one hand, the more discounts the platform offers, the more passengers are willing to join, and the higher the service rate will be. On the other hand, by increasing the discounts for passengers, the platform can first earn more revenue by attracting more passengers; the revenue will decrease if the discount exceeds a certain value, as the trip fare of passengers is too low to increase the total revenue by accommodating more passengers. In addition, high-capacity ride-sharing ($C = 4$) plays a more important role, achieving significantly more revenue than low-capacity ride-sharing ($C = 2$). It should be noted that since the demand is the same across all scenarios in a city, the average revenue of each vehicle is less in scenarios with 200 vehicles than in scenarios with 100 vehicles. But 200 vehicles can achieve more revenue in total than 100 vehicles.

As summarized in Table~\ref{tab: comp_perf}, this study further compares the best-performing fixed-discount benchmark with the heuristic algorithm (RINS) and the proposed joint optimization method (RGDS) in terms of average revenue per vehicle (AR) and service rate (SR). Overall, RGDS delivers the most consistent and substantial improvements. Specifically, compared with the fixed-discount strategy, RINS increases average revenue by 0.9--5.4\% (with a mean improvement of 2.7\%) and improves service rates by up to 17.6\% (with an average improvement of 8.1\%), although its performance gains are uneven and even slightly negative in the Shanghai scenario with $N=200$ and $C=2$. On the contrary, RGDS consistently yields larger revenue gains, ranging from 3.2\% to 8.7\% with a mean improvement of \textbf{5.2\%}, which exceeds that of RINS across all settings. These gains are particularly pronounced when the fleet size is small ($N=100$), highlighting the effectiveness of the RGDS algorithm under supply-constrained conditions. In terms of service rates, RGDS generally accommodates more passengers than both Fix and RINS. Relative to the benchmark, RGDS improves service rates by 4.5--20.0\% with an average increase of \textbf{8.2\%} in all scenarios, with the largest improvement observed in Chengdu with $N=100$ and $C=2$, where both fleet size and vehicle capacity are limited. While both RINS and RGDS experience slight service-rate degradation in the Shanghai case with $N=200$ and $C=2$, the decrease under RGDS is marginal (0.5\%) and smaller than that of RINS. Increasing vehicle capacity from $C=2$ to $C=4$ raises the baseline service rates for all methods; nevertheless, RGDS continues to provide consistent improvements in both revenue and service rate over Fix and RINS, demonstrating robustness across capacity levels. Overall, these results confirm that RGDS not only dominates the fixed-discount strategy but also improves upon the heuristic algorithm, RINS, by delivering more reliable and economically meaningful gains in profitability and operational efficiency across diverse market conditions.

\begingroup
\setlength{\tabcolsep}{4pt}
\renewcommand{\arraystretch}{1.2}
\begin{table}[!ht]
\caption{Comparison of average revenue (AR) and service rate (SR) among Fix, RINS, and RGDS pricing strategies across various scenarios.}
\begin{center}
\begin{tabular}{ l c c c c c c c c}
\toprule
& \multicolumn{4}{c}{Chengdu} & \multicolumn{4}{c}{Shanghai} \\
\midrule
& \multicolumn{2}{c}{N = 100} & \multicolumn{2}{c}{N = 200}
& \multicolumn{2}{c}{N = 100} & \multicolumn{2}{c}{N = 200} \\
\midrule
C = 2 & AR (\$) & SR (\%) & AR (\$) & SR (\%) & AR (\$) & SR (\%) & AR (\$) & SR (\%) \\
\midrule
Fix & 247.4 & 8.5 & 201.0 & 15.9 & 241.3 & 9.1 & 206.5 & 18.4 \\
RINS & 255.3 & 10.0 & 207.0 & 18.1 & 254.4 & 9.9 & 212.6 & 18.2 \\
Increase (\%) & 3.2 $\uparrow$ & 17.6 $\uparrow$ & 3.0 $\uparrow$ & 13.8 $\uparrow$ & 5.4 $\uparrow$ & 8.8 $\uparrow$ & 3.0 $\uparrow$ & -1.1 $\downarrow$ \\
RGDS & 269.0 & 10.2 & 210.2 & 17.8 & 262.3 & 9.8 & 214.5 & 18.3 \\
Increase (\%) & 8.7 $\uparrow$ & 20.0 $\uparrow$ & 4.6 $\uparrow$ & 11.9 $\uparrow$ & 8.7 $\uparrow$ & 7.7 $\uparrow$ & 3.9 $\uparrow$ & -0.5 $\downarrow$ \\
\midrule
C = 4 & AR (\$) & SR (\%) & AR (\$) & SR (\%) & AR (\$) & SR (\%) & AR (\$) & SR (\%) \\
\midrule
Fix & 277.3 & 10.8 & 214.0 & 18.3 & 269.8 & 10.3 & 219.0 & 19.9 \\
RINS & 282.7 & 12.1 & 218.6 & 19.2 & 275.0 & 11.1 & 221.0 & 20.1 \\
Increase (\%) & 1.9 $\uparrow$ & 12.0 $\uparrow$ & 2.1 $\uparrow$ & 4.9 $\uparrow$ & 1.9 $\uparrow$ & 7.8 $\uparrow$ & 0.9 $\uparrow$ & 1.0 $\uparrow$ \\
RGDS & 290.0 & 11.6 & 223.1 & 19.9 & 280.0 & 10.9 & 225.9 & 20.8 \\
Increase (\%) & 4.6 $\uparrow$ & 7.4 $\uparrow$ & 4.3 $\uparrow$ & 8.7 $\uparrow$ & 3.8 $\uparrow$ & 5.8 $\uparrow$ & 3.2 $\uparrow$ & 4.5 $\uparrow$ \\

\bottomrule
\end{tabular}
\label{tab: comp_perf}
\end{center}
\end{table}
\endgroup

\subsection{Service rates across various trip distances}
\label{sec: service_rate}

Figure~\ref{fig: SR_dist} reports passengers’ service rates as a function of trip distance under different fleet sizes and vehicle capacities in Chengdu and Shanghai. Several consistent patterns can be observed. First, for both cities and all scenarios, the proposed optimization method (Dynamic) improves service rates across almost the entire range of trip distances compared with the fixed-discount benchmark, and the improvement is particularly evident for distance groups whose service rates are below the global average. This indicates that the proposed method effectively mitigates systematic under-service of certain passenger groups rather than concentrating benefits on already well-served trips. Second, the magnitude and location of improvements vary with market conditions. For example, in Chengdu with a limited fleet ($N=100$), the proposed method mainly increases service rates for medium-distance trips (approximately 5--10~km); while with $N=200$, the service rates of long-haul requests (around 8--18~km) are significantly increased. In addition, increasing vehicle capacity from $C=2$ to $C=4$ raises the overall service rates for both methods, yet the proposed method continues to dominate the benchmark across distance bins, indicating that its effectiveness is robust to capacity changes. Importantly, the proposed method does not systematically favor long-distance passengers; instead, it improves service rates for short- and medium-distance trips as well, which are typically more difficult to pool efficiently. These findings suggest that the observed revenue gains are achieved primarily through improved traffic efficiency and better utilization of ride-sharing opportunities, rather than by prioritizing long-haul passengers, highlighting the fairness and practical applicability of the proposed method in real-world ride-sharing systems.

\begin{figure}[!ht]
    \centering
    \includegraphics[width = 0.98\textwidth]{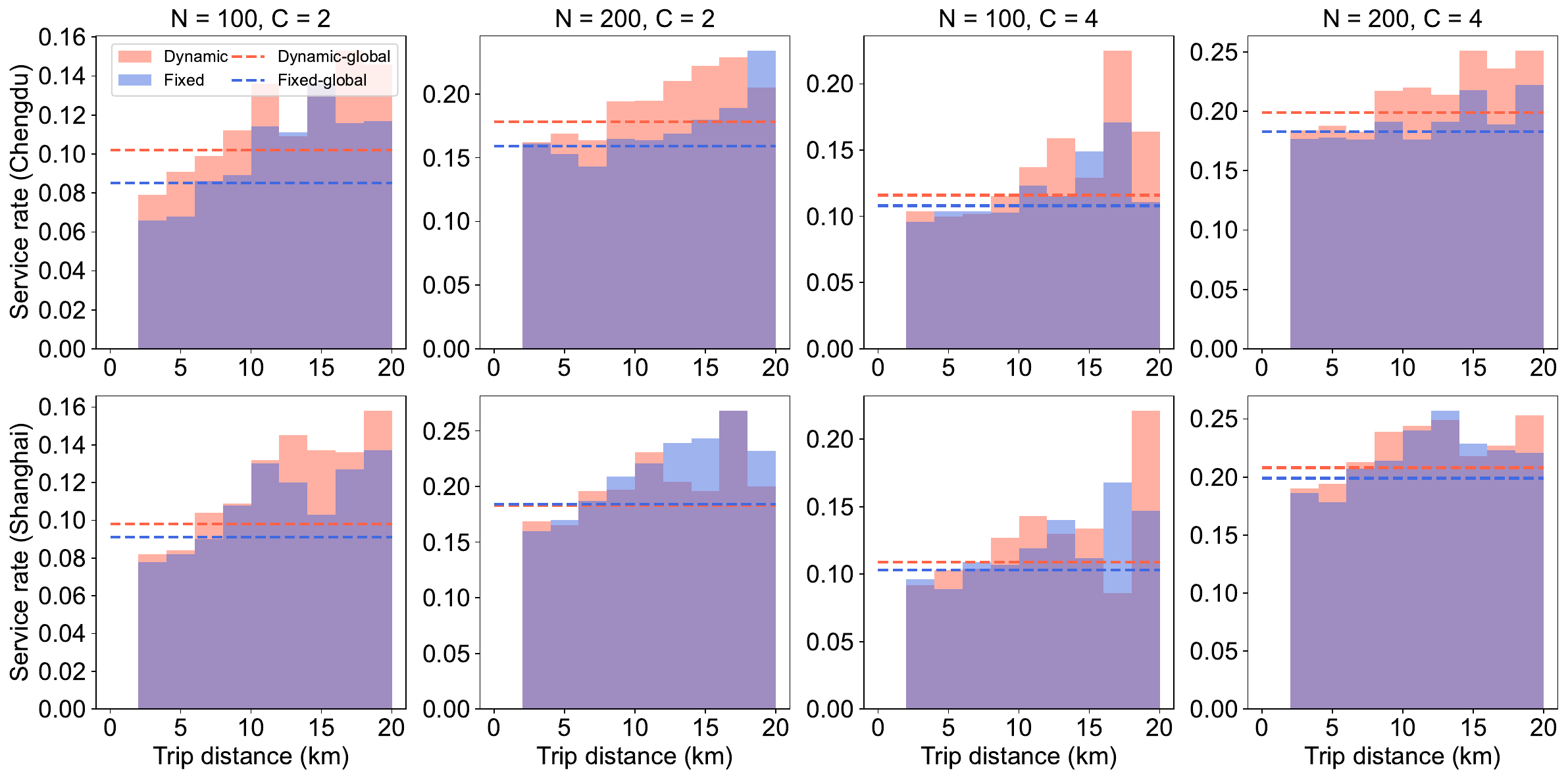}
    \caption{Service rates of passengers with different trip distances in Chengdu (the first line) and Shanghai (the second line) achieved by the benchmark (denoted as Fixed) and the proposed optimization method (denoted as Dynamic). Global denotes the global average service rate of all passengers.}
    \label{fig: SR_dist}
\end{figure}

\subsection{Spatiotemporal analysis}
\label{sec: spatiotemporal_analysis}

Figure~\ref{fig: Temporal} illustrates the temporal evolution of the average discounts offered by the proposed joint optimization method under different fleet sizes and vehicle capacities in Chengdu and Shanghai. Several observations can be drawn. First, for both cities, scenarios with a smaller fleet size ($N=100$) consistently result in lower discounts than those with $N=200$, reflecting higher demand pressure and limited supply, under which the system does not need to offer substantial discounts to incentivize passengers to accept ride-sharing. Second, the effect of vehicle capacity depends on market conditions. When $N=100$, vehicles with higher capacity ($C=4$) are associated with higher discounts than those with $C=2$, as the additional seating capacity allows the system to attract and accommodate more passengers through stronger price incentives. On the contrary, when $N=200$, the discounts offered under $C=2$ and $C=4$ are very similar and remain at relatively high levels, indicating that in supply-abundant markets the system must rely on higher discounts to stimulate demand regardless of vehicle capacity. Third, strong temporal patterns are observed in both cities. Before 17:00, average discounts are significantly higher, which is consistent with the markedly lower request arrival rates during this period; the system therefore offers higher discounts to encourage passenger participation in ride-sharing services. After approximately 18:30, the discount trajectories in the two cities diverge slightly: discounts increase in Chengdu but decrease in Shanghai, corresponding to the gradual decline and slight increase in arrival rates in the two cities, respectively, as shown in Figure~\ref{fig: demand_arrival_rate}. Overall, these results demonstrate that the proposed joint optimization algorithm can effectively capture temporal demand variations and structural differences across markets, dynamically adjusting discounts to improve system efficiency and increase platform revenue.

\begin{figure}[!ht]
    \centering
    \includegraphics[width = 0.85\textwidth]{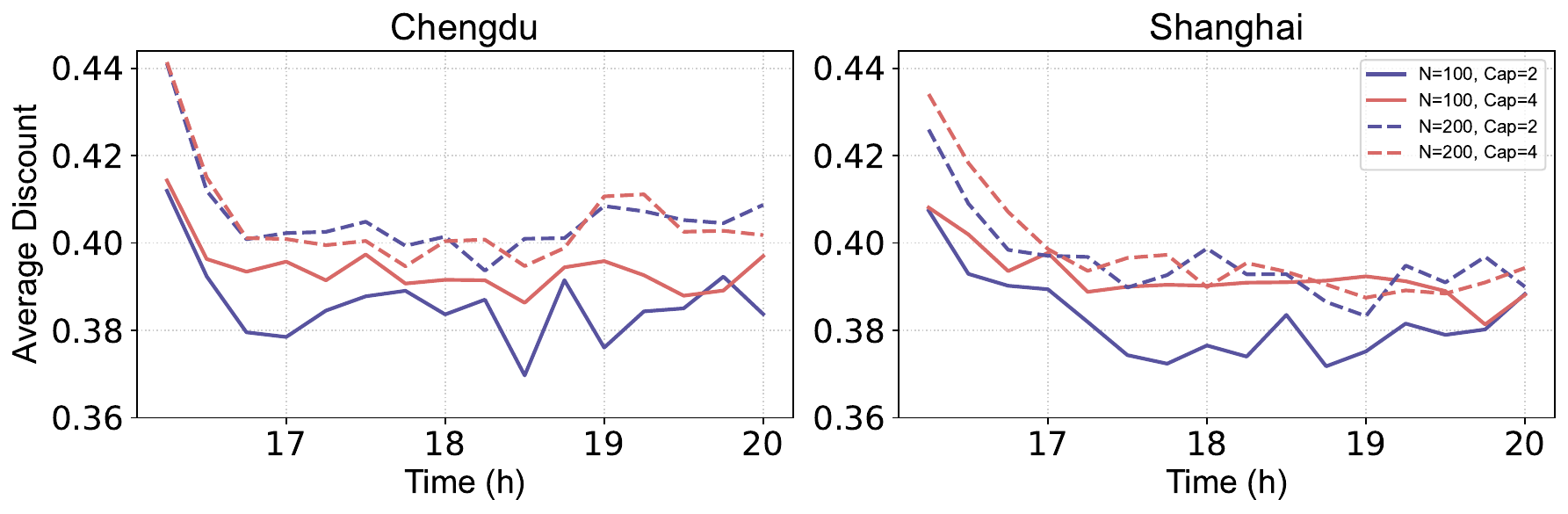}
    \caption{Temporal distribution of offered discounts.}
    \label{fig: Temporal}
\end{figure}

\begin{figure}[!ht]
    \centering
    \includegraphics[width = 0.85\textwidth]{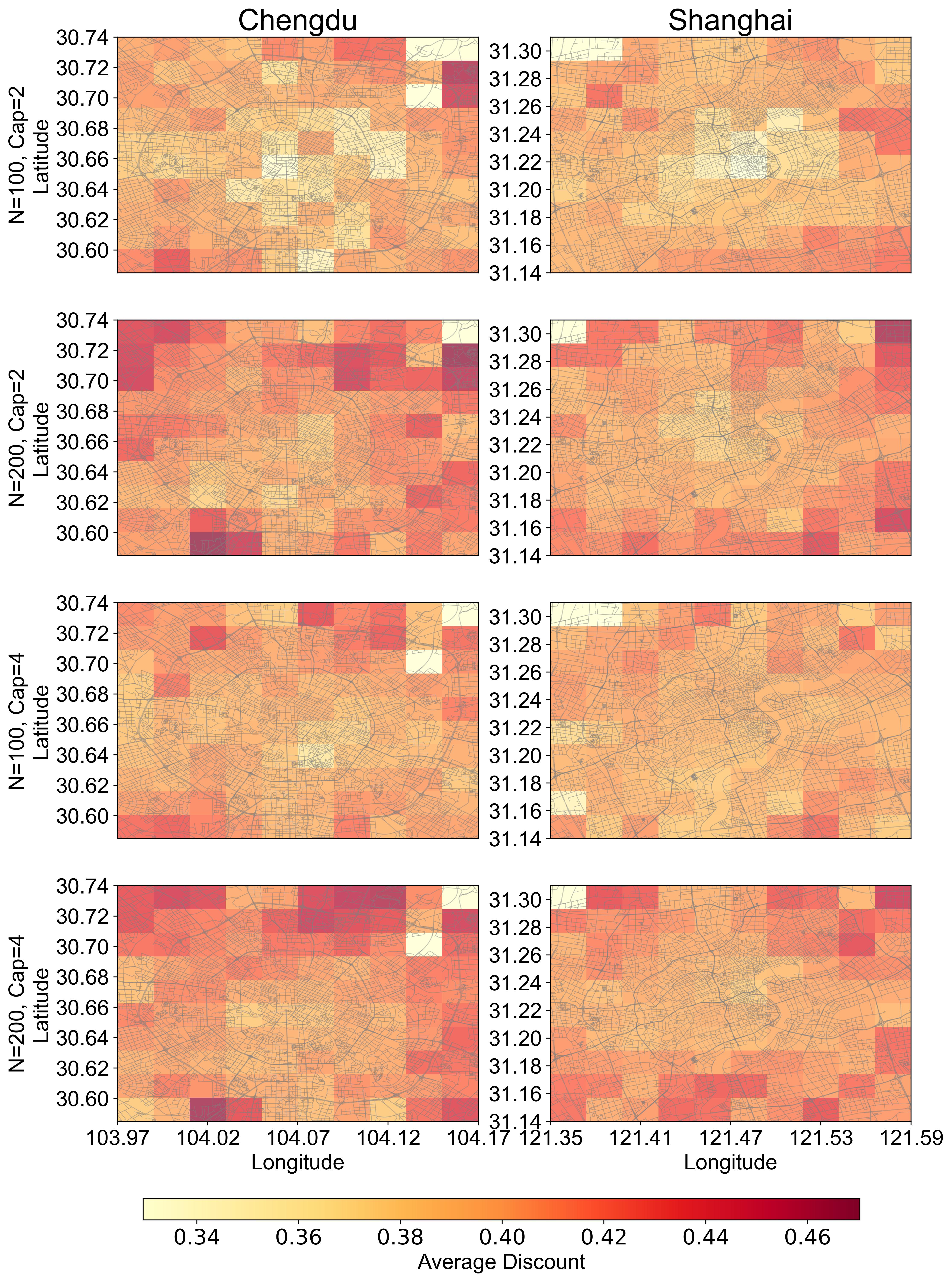}
    \caption{Spatial distribution and offered discounts.}
    \label{fig: Spatial}
\end{figure}

Figure~\ref{fig: Spatial} illustrates the spatial distribution of average discounts offered by the proposed joint optimization method under different fleet sizes and vehicle capacities in Chengdu and Shanghai. Also, several clear spatial and structural patterns emerge. First, in both cities, discounts are significantly lower in the central hot zones and noticeably higher in the surrounding and peripheral areas. This pattern is consistent with the spatial concentration of demand (as shown in Figure~\ref{fig: demand_spatial}): intense demand in city centers reduces the need for strong price incentives, whereas higher discounts are required in less active regions to stimulate ride-sharing participation. Second, scenarios with a larger fleet size ($N=200$) exhibit uniformly higher discounts than those with $N=100$, reflecting the need to actively attract additional passengers when supply is relatively sufficient. Third, increasing vehicle capacity from $C=2$ to $C=4$ leads to slightly higher discounts across most regions, as vehicles with larger capacities can accommodate more passengers and thus benefit from stronger incentives to promote pooling. Fourth, for the same configuration of fleet size and capacity, the average discounts in Chengdu are slightly higher than those in Shanghai, which is consistent with Chengdu’s marginally lower request arrival rate (as shown in Figure~\ref{fig: demand_arrival_rate}) and the corresponding need for greater incentives to balance supply and demand. Overall, these spatial patterns demonstrate that the proposed joint optimization method can effectively capture heterogeneous demand intensities across space and adapt pricing decisions to diverse operational scenarios, thereby improving market efficiency and enhancing platform revenue.

\subsection{Pooling size analysis}
\label{sec: pooling_size}

Figure~\ref{fig: pooling_size} illustrates the distribution of pooling sizes under three pricing strategies: Fix, RINS, and RGDS, where the pooling size is defined as the number of additional requests served together with a given answered request, and a value of zero indicates no successful pooling. As in previous analyses, the dynamic ride-sharing system allows requests to be pooled en route; therefore, a passenger may share a trip with multiple other passengers over the course of the journey even when the vehicle capacity is limited to $C=2$. Several common patterns can be observed across all strategies. When $C=2$, pooling is already effective: around 60\% of passengers are pooled with exactly one other passenger, and more than 20\% share trips with two or more passengers, demonstrating the effectiveness of dynamic en-route matching. When the vehicle capacity increases to $C=4$, the distribution clearly shifts toward larger pooling sizes, with the majority of passengers sharing their trips with one to three other passengers, reflecting the stronger accommodation ability of high-capacity ride-sharing. In all scenarios, the proportion of unpooled passengers remains around 10\%, suggesting that the pooling performance of the dynamic ride-sharing system is robust across cities and fleet sizes.

\begin{figure}[!ht]
    \centering
    \includegraphics[width = 1.0\textwidth]{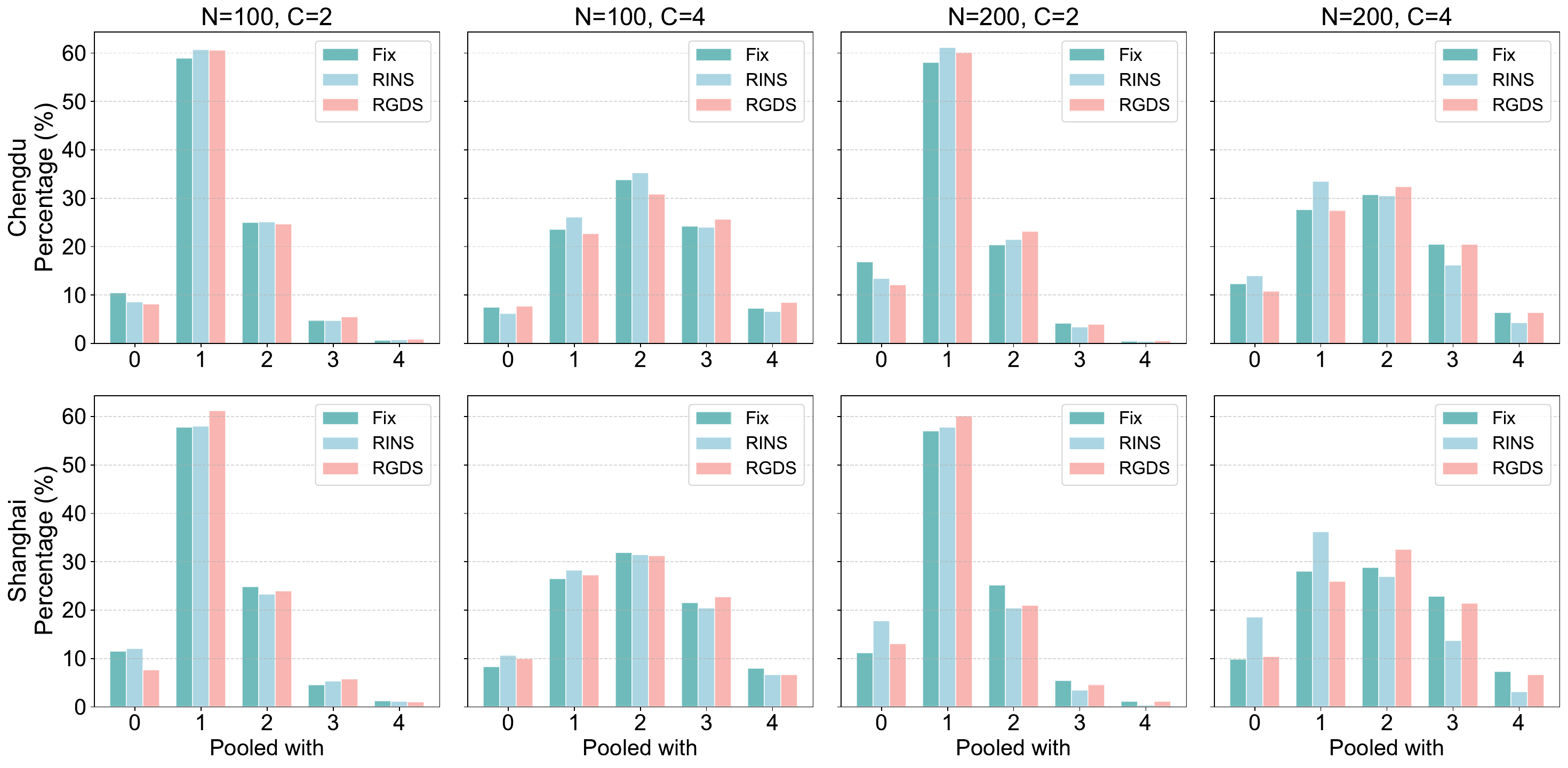}
    \caption{Distribution of pooling sizes.}
    \label{fig: pooling_size}
\end{figure}

Despite these common trends, systematic differences emerge among the three strategies. As shown in Figure~\ref{fig: pooling_size}, RGDS consistently yields a higher proportion of passengers with larger pooling sizes, particularly for pooling sizes of two and three, while reducing the share of unpooled or weakly pooled trips. On the contrary, RINS exhibits pooling-size distributions that are comparable to or slightly worse than those of the fixed pricing strategy, with more mass concentrated on smaller pooling sizes. This is because RINS relies on local and coordinate-wise searches that focus on improving the expected revenue of individual pricing decisions. While this approach can increase revenue and service rate by enhancing fleet utilization, it largely overlooks the global pricing and matching interactions that are critical for effective ride-sharing.

Table~\ref{tab: pooling_size} further provides a quantitative comparison of average pooling size and pooling distance across various scenarios among the three strategies. First, increasing the vehicle capacity from $C=2$ to $C=4$ substantially enlarges both the average pooling size and pooling distance for all methods, confirming that higher-capacity vehicles enable more intensive ride-sharing and longer shared travel segments. Second, clear performance differences across strategies can be observed. Compared with Fix, RGDS consistently achieves the largest pooling sizes and longest pooling distances under both capacity settings, increasing the average pooling size from 1.24 to 1.26 and the pooling distance from 4.79~km to 4.93~km when $C=2$, and maintaining a pooling size of 2.02 while extending the pooling distance to 6.26~km when $C=4$. However, RINS results in smaller pooling sizes and shorter pooling distances than both Fix and RGDS, with average pooling sizes of 1.20 and 1.79 and pooling distances of 4.59~km and 5.61~km for $C=2$ and $C=4$, respectively. These findings highlight a fundamental difference between local and global optimization strategies. RINS relies on a coordinate-wise search mechanism that prioritizes short-term revenue improvements through localized pricing adjustments, while overlooking the global pricing–matching interactions required to enhance pooling efficiency. On the contrary, RGDS explicitly accounts for global gradients and jointly optimizes pricing and matching decisions across the system, enabling vehicles to pool more requests over longer distances. This global perspective translates into higher pooling efficiency, improved vehicle utilization, and superior system-level performance.

\begingroup
\setlength{\tabcolsep}{4pt}
\renewcommand{\arraystretch}{1.2}
\begin{table}[!ht]
\caption{Comparison of pooling sizes and distances across various scenarios between the benchmark and joint optimization algorithm.}
\begin{center}
\begin{tabular}{c c c c c}
\toprule
& \multicolumn{2}{c}{C = 2} & \multicolumn{2}{c}{C = 4} \\
\midrule
 & Pooling size & Pooling distance (km) &  Pooling size & Pooling distance (km) \\
\midrule
Fix & 1.24 & 4.79  & 2.02 & 6.16 \\
RINS & 1.20 & 4.59 & 1.79 & 5.61 \\
RGDS & 1.26 & 4.93 & 2.02 & 6.26 \\
\bottomrule
\end{tabular}
\label{tab: pooling_size}
\end{center}
\end{table}
\endgroup

\section{Conclusion}\label{sec: conclusion}

This study proposes an uncertainty-aware joint optimization framework for dynamic high-capacity ride-sharing. The framework adopts a two-stage stochastic program to effectively take into account passengers' choice uncertainty, which can simultaneously optimize upfront pricing and passenger-vehicle matching. Also, the proposed RGDS algorithm can solve the stochastic program with significantly lower computational costs while achieving high accuracy. In addition, the proposed method can be applied to high-capacity dynamic ride-sharing without significantly increasing computation time, drastically expanding its application prospects. The numerical experiments in Chengdu and Shanghai demonstrate that the proposed joint optimization framework can bring valuable improvement to dynamic ride-sharing, increasing the average revenue by 5.2\% and service rate by 8.2\% across various scenarios. The proposed framework provides a valuable reference to TNCs, informing them how to design reasonable upfront prices and operate vehicles efficiently.

There are a few limitations that should be acknowledged in this study. First, this study focuses on dynamic ride-sharing services and does not consider other transportation modes. Also, the questionnaire survey involves only passengers' willingness to participate in ride-sharing and only considers discounts and detour distances. Nevertheless, this study provides insights for TNCs, especially in the era of autonomous driving, who may operate a fleet of autonomous vehicles and would like to seize the market share of dynamic ride-sharing. In this context, the proposed method can be used to promote dynamic ride-sharing, increasing the total revenue and enhancing passengers' experience.

In the future, it is interesting to extend this study to multi-mode service scenarios, where a platform provides both ride-sharing and solo-hailing services. Correspondingly, a more detailed survey can be conducted to determine passengers' choices considering the waiting time, detour time, and price. In addition, the dynamics and uncertainty of ride-sharing, e.g., future sharing probability, can be captured by advanced algorithms, which may further improve the performance of dynamic ride-sharing. Also, considering traffic congestion, vehicles' routes can be finely designed to improve the probability that passengers can share their trips and reduce their detour time.

\section*{Declaration of competing interest}

The authors declare that they have no known competing financial interests or personal relationships that could have appeared to
influence the work reported in this paper.

\section*{Acknowledgement}

This work was supported by the General Research Fund (GRF) of the Research Grants Council of Hong Kong under Project No. HKU17204525.


\newpage

\bibliography{reference}

\appendix

\newpage

\section{Choice uncertainty of passengers}

\label{seca: survey}

To quantify passengers’ willingness to accept ride pooling under different service conditions, we conduct a questionnaire survey across multiple major cities, including Shenzhen, Chengdu, and Hong Kong. The survey focuses on passengers’ tolerance for additional travel time (detour) in exchange for fare discounts in ride-sharing systems. Specifically, respondents are asked to state the minimum expected discount they would require under different detour levels, defined as 10\%, 20\%, 30\%, 40\%, and 50\% increases relative to the original travel distances without pooling. For each detour level, the minimum acceptable discount is selected from seven ordered categories: below 5\%, 5--10\%, 10--15\%, 15--20\%, 20--30\%, 30--40\%, above 40\%, as well as an option indicating that the respondent would never choose ride pooling. In total, 403 questionnaires are collected. To ensure internal consistency and rationality of responses, we apply a filtering rule based on monotonic preference: for a given respondent, the required minimum discount should be non-decreasing as the detour increases. Questionnaires that violate this condition are deemed invalid and removed. After filtering, 371 effective questionnaires remain and are used for subsequent analysis.

\begin{figure}[!ht]
    \centering
    \includegraphics[width = 0.8\textwidth]{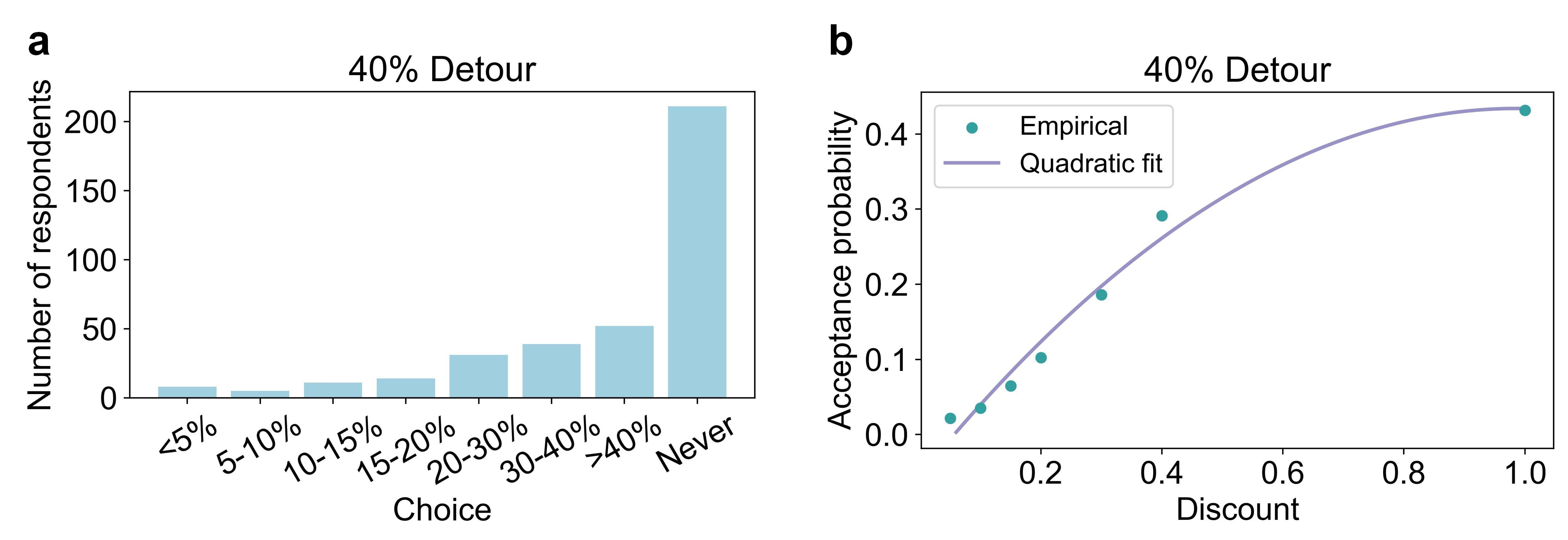}
    \caption{Survey results with a detour time of 40\%. a, Respondents' choice distribution. b, Fitting result using a quadratic function.}
    \label{figa: choice}
\end{figure}

Using the effective questionnaires, we model passengers’ acceptance behavior in terms of the probability of accepting ride pooling given a discount level. Each respondent’s answer is first converted into a discount threshold, representing the minimum discount required for acceptance. For a given discount level, a passenger is considered to accept pooling if the offered discount is no less than this threshold. The empirical acceptance probability is then computed as the fraction of passengers who would accept pooling at that discount level. To reduce computational burden in the simulation-based experiments, we focus on the detour level of 40\%, which represents a moderate yet practically relevant level of service degradation. The choice distribution at 40\% detour is shown in Figure~\ref{figa: choice}(a). For this detour level, the empirical acceptance probabilities are evaluated on a discrete discount grid, including 5\%, 10\%, 15\%, 20\%, 30\%, 40\%, and 100\%.

To obtain a smooth and tractable representation of passenger acceptance behavior, as shown in Figure~\ref{figa: choice}(b), we approximate the empirical acceptance probabilities using a quadratic function. Specifically, the acceptance probability $G(\theta)$ is modeled as
\[
G(\theta) = a \theta^2 + b \theta + c,
\]
where $\theta$ denotes the discount level expressed as a fraction, and $a$, $b$, and $c$ are parameters. The quadratic coefficients are estimated via least squares. For the 40\% detour case, the fitted acceptance function is given by
\[
G(\theta) = -0.5017\, \theta^2 + 0.9900\,\theta - 0.0544,
\]
with a coefficient of determination $\mathrm{R}^2 = 0.981$, indicating an excellent fit to the empirical data. This result suggests that passengers’ acceptance probability increases at a decelerating rate as the discount becomes larger, highlighting the nonlinear sensitivity of pooling acceptance to monetary incentives. The estimated acceptance function is used in all numerical experiments to model passenger choice behavior under ride pooling.

\section{RINS algorithm}
\label{seca: RINS}

The relaxation--induced neighborhood search (RINS) algorithm~\citep{danna2005exploring} adopted in this study is formulated in Algorithm~\ref{alg: RINS}.

\begin{algorithm}[!ht]
\caption{RINS (Relaxation-Induced Neighborhood Search)}
\label{alg: RINS}
\SetAlgoLined
\KwIn{Vehicle set $K$; trip set $F$; assignment utilities $u$; step size $\gamma > 0$; minimum improvement threshold $\epsilon > 0$; maximum total number of search steps $S \in \mathbb{N}$; maximum number of steps per coordinate $S_c \in \mathbb{N}$.}
\KwOut{Upfront discount vector $\bm{\theta}$ for new passengers.}

Obtain initial solution $\bm{\theta}^0$ by solving the RP-CO \;
Obtain the initial objective value $J^0$ by solving the RAP-RS \;
Sort the indices of $\bm{\theta}^0$ in descending order of their values \;
$\bm{\theta} \gets \bm{\theta}^0$, $J \gets J^0$, $s \gets 0$ \;

\ForEach{index $i$ in sorted order}{
    Set step direction $d \gets \gamma$ and $s_i \gets 0$ \;
    \While{$s_i < S_c$ and $s < S$}{
        $\bm{\theta}_i \gets \bm{\theta}_i + d$ \;
        Obtain new objective value $J'$ by solving the RAP-RS \;
        $s \gets s + 1$, $s_i \gets s_i + 1$ \;

        \eIf{$J' < J$}{
            Compute relative improvement $\Delta = |(J - J') / J|$ \;
            $J \gets J'$ \;
            \If{$\Delta < \epsilon$}{
                \textbf{break} \tcp*{Stop searching this coordinate}
            }
        }{
            Revert $\bm{\theta}_i \gets \bm{\theta}_i - d$ \;
            \If{$s_i = 1$}{
                $d \gets -d$ \tcp*{Flip direction once}
            }
            \Else{
                \textbf{break} \tcp*{Give up on this coordinate}
            }
        }
    }
    \If{$s \geq S$}{
        \textbf{break}
    }
}
\Return $\bm{\theta}$ \;
\end{algorithm}

\section{Proof of Theorem 1}
\begin{proof}
    For any fixed scenario $\xi$, the second-stage problem selects a set of trips (hyperedges) so that each request belongs to at most one selected trip, and maximizes the total utility. Each trip $f\in F^\xi$ contains at most $p$ requests, so this is exactly the (weighted) $p$-set packing problem, which is NP-hard \citep{furedi1993fractional, chan2012linear, luo2023efficient}. Hence, the second-stage problem is NP-hard. The full SP consists of many such second-stage problems, coupled with first-stage decisions. Since it already contains an NP-hard subproblem as a special case, it is itself strongly NP-hard.
\end{proof}

\section{Proof of Theorem 2}

\begin{proof}
We prove Theorem~\ref{thm: equivalence} by first characterizing the objective of the two-stage stochastic program (SP) as a function of $\boldsymbol\theta$, and then doing the same for the deterministic relaxation (RP), finally showing that the two coincide.

\medskip
\noindent\textbf{Step 1: Reduction of SP to a deterministic problem.}

For each passenger $r\in R_n$ and scenario $\xi$, let
\[
  y_r^\xi =
  \begin{cases}
    1, & \text{if passenger $r$ accepts the offer in scenario $\xi$},\\[0.2em]
    0, & \text{otherwise},
  \end{cases}
\]
with
\[
  \mathbb{P}(y_r^\xi = 1 \mid \theta_r) = G(\theta_r), \qquad
  \mathbb{P}(y_r^\xi = 0 \mid \theta_r) = 1 - G(\theta_r),
\]
and assume independence across $r$.

Given $\boldsymbol\theta$ and a realization $\xi$, denote by $R_n^\xi := \{r\in R_n : y_r^\xi = 1\}$ the set of accepting passengers. Furthermore, let $R^\xi := R_n^\xi \cup R_w$ denote all passengers waiting to be matched. The second--stage problem in scenario $\xi$ is to choose binary assignment variables $x_{kr}^\xi \in \{0,1\}$ (vehicle $k$ serves passenger $r$) to maximize total utility:
\begin{align*}
  Q(\boldsymbol\theta,\xi)
  \;=\;
  \max_{x^\xi} 
  &\quad \sum_{r\in R^\xi} \sum_{k\in K} x_{kr}^\xi \, p_r(\theta_r)\\
  \text{s.t.}\ 
  &\quad \sum_{k\in K} x_{kr}^\xi \le 1, && \forall r\in R^\xi,\\
  &\quad \sum_{r\in R^\xi} x_{kr}^\xi \le 1, && \forall k\in K, \\
  &\quad x_{kr}^\xi \in \{0,1\}, && \forall k\in K,\ r\in R^\xi.
\end{align*}
Recall that $p_r(\theta_r) = u_{kr}(\theta_r)$ as we assume the pickup cost can be negligible. Since we assume that the supply is sufficient, there is no effective competition for vehicles: in any scenario $\xi$, it is feasible to assign each accepting passenger in $R^\xi$ to some vehicle while satisfying the one-passenger-per-vehicle constraints. In addition, $p_r(\theta_r)>0$ for all $r$ and feasible $\theta_r$, so it is always optimal in the second stage to serve every accepting passenger. Thus, for any fixed $\boldsymbol\theta$ and scenario $\xi$,
\begin{equation}
  Q(\boldsymbol\theta,\xi)
  \;=\;
  \sum_{r\in R^\xi} p_r(\theta_r)
  \;=\;
  \sum_{r\in R_n} y_r^\xi \, p_r(\theta_r) + Q_w,
  \label{eq:Q-SP}
\end{equation}
where $Q_w = \sum_{r\in R_w}p_r$ denotes the revenue from the awaiting passengers.

The first stage chooses $\boldsymbol\theta$ to maximize the expected value of $Q(\boldsymbol\theta,\xi)$:
\[
  \max_{\boldsymbol\theta} \ \mathbb{E}_\xi\bigl[Q(\boldsymbol\theta,\xi)\bigr]
  \quad \text{s.t.}\quad
  \theta^L \le \theta_r \le \theta^U,\ \forall r\in R_n.
\]
Using Equation~\eqref{eq:Q-SP} and linearity of expectation,
\begin{align*}
  \mathbb{E}_\xi\bigl[Q(\boldsymbol\theta,\xi)\bigr]
  &= \mathbb{E}_\xi\left[\sum_{r\in R_n} y_r^\xi \, p_r(\theta_r) + Q_w\right] \\
  &= \sum_{r\in R_n} \mathbb{E}_\xi\bigl[y_r^\xi\bigr]\, p_r(\theta_r) + Q_w \\
  &= \sum_{r\in R_n} G(\theta_r)\, p_r(\theta_r) + Q_w.
\end{align*}
Hence SP is equivalent to the deterministic program
\begin{equation}
  \max_{\boldsymbol\theta}
  \;\;
  \sum_{r\in R_n} G(\theta_r)\, p_r(\theta_r)
  \quad\text{s.t.}\quad
  \theta^L \le \theta_r \le \theta^U,\ \forall r\in R_n.
  \tag{SP$'$}
  \label{prob:SP-prime}
\end{equation}

\medskip
\noindent\textbf{Step 2: Reduction of RP to the same deterministic problem.}

In the deterministic relaxation (RP), the scenario-dependent binary assignment variables are replaced by deterministic continuous variables $\tilde x_{kr} \in [0,1]$, representing the expected or fractional assignment of vehicle $k$ to passenger $r$. The objective is
\[
  \max_{\tilde x, \boldsymbol\theta}
  \;\;
  \sum_{r\in R}\sum_{k\in K} \tilde x_{kr}\, p_r(\theta_r),
\]
subject to
\begin{align*}
  \sum_{k\in K} \tilde x_{kr} &\le G(\theta_r), && \forall r\in R_n, \\
  \sum_{k\in K} \tilde x_{kr} &\le 1, && \forall r\in R_w, \\
  \sum_{r\in R} \tilde x_{kr} &\le 1, && \forall k\in K, \\
  0 \le \tilde x_{kr} &\le 1, && \forall k\in K,\ r\in R, \\
  \theta^L \le \theta_r &\le \theta^U, && \forall r\in R.
\end{align*}
In the objective, each $\tilde x_{kr}$ appears with a nonnegative coefficient $p_r(\theta_r) > 0$. Thus, at optimality, it is never beneficial to leave slack in
\[
  \sum_{k\in K} \tilde x_{kr} \le G(\theta_r), \quad \forall r \in R_n,
\]
\[
  \sum_{k\in K} \tilde x_{kr} \le 1, \quad \forall r \in R_w,
\]
because increasing $\sum_k \tilde x_{kr}$ (while remaining feasible) strictly increases the objective. Since we have enough vehicles to allocate up to $G(\theta_r)$ (or one) units of assignment to each passenger $r, \, \forall r \in R_n$ (or $\forall r \in R_w$) without violating the vehicle constraints. Hence, there exists an optimal solution $(\tilde x^\star,\boldsymbol\theta)$ with
\begin{align}
\label{eq:sumx-equals-G}
  \sum_{k\in K} \tilde x_{kr}^\star &= G(\theta_r),
  & \forall r\in R_n,\\
  \sum_{k\in K} \tilde x_{kr}^\star &= 1,
  & \forall r\in R_w.
  \label{eq:sumx-equals-1}
\end{align}
Substituting Equations~\eqref{eq:sumx-equals-G} and~\eqref{eq:sumx-equals-1} into the objective gives
\begin{align*}
  \sum_{r\in R}\sum_{k\in K} \tilde x_{kr}^\star\, p_r(\theta_r)
  &= \sum_{r\in R} \left(\sum_{k\in K} \tilde x_{kr}^\star\right)\, p_r(\theta_r) \\
  &= \sum_{r\in R_n} G(\theta_r)\, p_r(\theta_r) + Q_w.
\end{align*}
Crucially, this expression depends only on $\boldsymbol\theta$ and not on the particular distribution of $\tilde x_{kr}^\star$ across $k$. Thus, for any feasible $\boldsymbol\theta$, the optimal value of RP as a function of $\boldsymbol\theta$ is exactly
\[
  \sum_{r\in R_n} G(\theta_r)\, p_r(\theta_r),
\]
and RP is equivalent (w.r.t.\ $\boldsymbol\theta$) to
\begin{equation}
  \max_{\boldsymbol\theta}
  \;\;
  \sum_{r\in R_n} G(\theta_r)\, p_r(\theta_r)
  \quad\text{s.t.}\quad
  \theta^L \le \theta_r \le \theta^U,\ \forall r\in R_n.
  \tag{RP$'$}
  \label{prob:RP-prime}
\end{equation}

\medskip
\noindent\textbf{Step 3: Equality of SP and RP objectives.}

Comparing \eqref{prob:SP-prime} and \eqref{prob:RP-prime}, we see that under the assumptions,
\begin{itemize}
  \item SP and RP have the same feasible set for $\boldsymbol\theta$:
  \(
    \theta^L \le \theta_r \le \theta^U,\ \forall r\in R_n,
  \)
  and
  \item For any feasible $\boldsymbol\theta$, the objective value in SP and in RP coincides and equals
  \[
    \sum_{r\in R_n} G(\theta_r)\, p_r(\theta_r).
  \]
\end{itemize}
Therefore, SP and RP are equivalent with respect to upfront discount optimization, as claimed.
\end{proof}

\section{Proof of Proposition 1}

\begin{proof}
We prove Proposition~\ref{prop: RP_not_convex} by showing that the objective function is not convex in the decision variables $(\tilde{x},\bm{\theta})$.

Rewriting Equation \eqref{eq: RP_obj} we get
\begin{align*}
    \min_{\tilde{x}, \bm{\theta}}\;
    &\sum_{f \in F}\sum_{k\in K}  \sum_{r \in f}
    \bigl[- (\eta_1 + \eta_2 l_r - \delta \cdot l_{kr})\tilde{x}_{kf}
        + (\eta_1 + \eta_2 l_r)\,\tilde{x}_{kf}\theta_{r}\bigr].
\end{align*}
Since
\[
    \sum_{f \in F}\sum_{k\in K}  \sum_{r \in f}
    - (\eta_1 + \eta_2 l_r - \delta \cdot l_{kr})\tilde{x}_{kf}
\]
is linear, it is convex. Thus, the convexity of the objective function is entirely determined by the bilinear term
\[
    \sum_{f \in F}\sum_{k\in K}  \sum_{r \in f}
    (\eta_1 + \eta_2 l_r)\,\tilde{x}_{kf}\theta_r.
\]

We define
\[
    g(x,\theta) := a\,x\theta.
\]
We now show that $g(x,\theta)=a\,x\theta$ is not convex on any domain that contains at least two distinct points. The Hessian of $g$ is
\[
    \nabla^2 g(x,\theta)
    =
    \begin{pmatrix}
        \dfrac{\partial^2 g}{\partial x^2} &
        \dfrac{\partial^2 g}{\partial x \partial \theta} \\[0.8em]
        \dfrac{\partial^2 g}{\partial \theta \partial x} &
        \dfrac{\partial^2 g}{\partial \theta^2}
    \end{pmatrix}
    =
    \begin{pmatrix}
        0 & a \\
        a & 0
    \end{pmatrix}.
\]
The eigenvalues of this $2\times 2$ matrix are
\[
    \lambda_1 = a > 0,\qquad \lambda_2 = -a < 0.
\]
Hence, the Hessian is indefinite and therefore not positive semidefinite. Consequently, $g$ is neither convex nor concave on $\mathbb{R}^2$ (or on any convex set of dimension 2 that contains points with both $x\neq 0$ and $\theta\neq 0$). Equivalently, the Hessian of the RP objective is not positive semidefinite, confirming that the objective is not convex. Since the constraints are linear (hence convex), the nonconvexity of the objective function implies that the relaxed program (RP) as a whole is not a convex optimization problem.
\end{proof}

\section{Proof of Proposition 2}

\begin{proof}
Convexity of an optimization problem requires:
\begin{enumerate}
\item a convex objective function, and
\item a convex feasible region (\emph{i.e.}, all constraints define a convex set).
\end{enumerate}
We verify both for RP-CO to prove Proposition~\ref{prop: RP-CO_convex}.

The RP-CO objective is
\[
   \min_{\tilde{\mathbf x},\bm{\theta},\mathbf w}
   \sum_{f\in F}\sum_{k\in K}\sum_{r\in f}
   -\left[(\eta_1+\eta_2 l_r-\delta l_{kr})\tilde x_{kf}
          -(\eta_1+\eta_2 l_r) w_{kf}^r\right].
\]
This is an affine (linear plus constant) function of the variables $(\tilde x_{kf})_{k,f}$ and $(w_{kf}^r)_{k,f,r}$. Hence, the objective is convex.

All four McCormick constraints for each $(k,f,r)$,
\begin{align*}
  w_{kf}^r &\ge \theta^L \tilde x_{kf}, \\
  w_{kf}^r &\ge \theta_r + \theta^U \tilde x_{kf} - \theta^U, \\
  w_{kf}^r &\le \theta_r + \theta^L \tilde x_{kf} - \theta^L, \\
  w_{kf}^r &\le \theta^U \tilde x_{kf},
\end{align*}
are linear inequalities in the decision variables $(\tilde x_{kf},\theta_r,w_{kf}^r)$. The remaining constraints \eqref{con: RP_C1}--\eqref{con: RP_C5} are those of RP, which are also linear (by construction of RP as a linear relaxation of the stochastic program). Therefore, the feasible region of RP-CO is convex.
\end{proof}

\section{Proof of Proposition 3}

\begin{proof}
We compare the variable sets and constraints in the two models to prove Proposition~\ref{prop: feasible}.

In the two-stage stochastic program (SP), the first-stage decision variables are the continuous discount variables $\bm{\theta}$, and they are subject to linear constraints ($\theta^L\le\theta_r \le\theta^U$). While RP-CO uses $(\tilde{\mathbf x}, \bm{\theta}, \mathbf w)$, where $\bm{\theta}$ are exactly the same discount variables as in the first stage of SP, with the same lower and upper bounds; and $\tilde{\mathbf x}$ and $\mathbf w$ are auxiliary continuous variables that appear only in constraints and objective of RP-CO.

The constraints \eqref{con: RP_C5} in RP-CO coincide with the corresponding first-stage constraints in SP. The McCormick constraints are purely auxiliary: they link $w_{kf}^r$, $\tilde x_{kf}$, and $\theta_r$ to approximate the bilinear product $\tilde x_{kf}\theta_r$. These constraints restrict how $w_{kf}^r$ may vary relative to $(\tilde x_{kf},\theta_r)$ but do not change the admissible set of $\bm{\theta}$ relative to SP.

Let $(\tilde{\mathbf x}^\star,\bm{\theta}^\star,\mathbf w^\star)$ be an optimal solution to RP-CO. Then, the variables $\bm\theta^\star$ satisfy constraints \eqref{con: RP_C5}, which are exactly the first-stage constraints of SP ($\theta^L\le\theta_r \le\theta^U$). Consequently, the initial solution obtained from RP-CO induces a feasible first-stage solution for SP.
\end{proof}

\section{Proof of Proposition 4}
\begin{proof}
    We prove Proposition~\ref{prop: RSP_AS_convex} by showing that both the objective and constraints in RAP-AS are convex.

    For fixed discounts \(\bm\theta\), the quantities \(P(\xi)\), \(u_{kf}^\xi\), and \(y_r^\xi\) are constants in each scenario. Thus, the objective is
    \[
        J(\hat{x}) = \sum_{\xi,f,k} c_{kf}^\xi \hat{x}_{kf}^\xi,
        \qquad c_{kf}^\xi := -P(\xi) u_{kf}^\xi,
    \]
    which is an affine function of the decision variables \(\hat{x}_{kf}^\xi\). Hence, the objective is convex.
    
    Regarding the constraints, first, the bounds \(\hat{x}_{kf}^\xi \in [0,1]\) are linear inequalities (\(0 \le \hat{x}_{kf}^\xi \le 1\)). Also, scenario-wise assignment constraints \eqref{con: SP_C1}--\eqref{con: SP_C3}:
      \begin{align*}
        &\sum_{f \in F^\xi} \hat{x}_{kf}^\xi \le 1, &&\forall k \in K,\xi \in \Xi, \\
        &\sum_{k\in K}\sum_{f\in F^\xi: r\in f}\hat{x}_{kf}^\xi \le y_r^\xi, &&\forall r\in R_n,\xi \in \Xi, \\
        &\sum_{k\in K}\sum_{f\in F^\xi: r\in f}\hat{x}_{kf}^\xi \le 1, &&\forall r\in R_w,\xi \in \Xi.
      \end{align*}
    For fixed \(\bm\theta\), the right-hand sides \(y_r^\xi\) are constants (0 or 1). Hence, all these constraints are linear in \(\hat{x}\). Consequently, the feasible region is a convex set.
\end{proof}

\section{Proof of Theorem 3}

\begin{proof}
    We prove Theorem~\ref{thm: J = O} by showing that RAP-AS has an optimal integral solution, which is equivalent to that of the original SP.

\textbf{Step 1: Scenario-wise decoupling.}

For a fixed $\bm{\theta}$, the decision variables $\hat{x}^\xi$ are not coupled across scenarios by any constraints; each scenario $\xi$ has its own copy of
constraints \eqref{con: SP_C1}--\eqref{con: SP_C3}. Thus, RAP-AS decomposes
into independent subproblems, one per scenario:
\[
 \min_{\hat{x}^\xi}
 \sum_{f\in F^\xi}\sum_{k\in K} -P(\xi)\,\hat{x}_{kf}^\xi u_{kf}^\xi
 \quad\text{s.t.}\quad
 \hat{x}_{kf}^\xi\in[0,1],\ \text{constraints } \eqref{con: SP_C1}\textbf{--}\eqref{con: SP_C3}.
\]
Solving RAP-AS is equivalent to solving each scenario subproblem and summing the
weighted optimal values.

\textbf{Step 2: Integrality of each scenario's relaxation.}

Consider a fixed scenario $\xi$ and drop the index $\xi$ for brevity. The constraints are
\begin{align*}
 &\sum_{f \in F} x_{kf} \le 1, &&\forall k\in K,\\
 &\sum_{k \in K}\sum_{f \in F: r\in f} x_{kf} \le y_r, &&\forall r\in R_n,\\
 &\sum_{k \in K}\sum_{f \in F: r\in f} x_{kf} \le 1, &&\forall r\in R_w,\\
 &0 \le x_{kf} \le 1 &&\forall k \in K, f \in F.
\end{align*}
Here, $y_r\in\{0,1\}$ is fixed data for this scenario. Each variable $x_{kf}$ represents assigning vehicle $k$ to trip $f$, and each constraint limits a single resource (vehicle or passenger) to be used by at most one trip. The coefficient matrix of such a system is totally unimodular. Since all right-hand sides (including $y_r$) are integral, the resulting polyhedron
\[
   H := \{x \in \mathbb{R}^{|K||F^\xi|} :
        0\le x\le 1, \text{constraints \eqref{con: SP_C1}--\eqref{con: SP_C3} } \}
\]
has only integral extreme points. Therefore, the relaxation ($x_{kf}\in[0,1]$ instead of $x_{kf}\in\{0,1\}$) has an optimal solution that is integral, \emph{i.e.}, $x_{kf}\in\{0,1\}$ for all $(k,f)$. Hence, for each scenario $\xi$, the optimal objective value of the relaxed subproblem equals that of the original integer subproblem:
\[
   \min_{\hat{x}^\xi \in [0,1]} \sum_{f,k} -P(\xi)\,\hat{x}_{kf}^\xi u_{kf}^\xi
   \;=\;
   \min_{x^\xi\in\{0,1\}} \sum_{f,k} -P(\xi)\,x_{kf}^\xi u_{kf}^\xi
   = P(\xi)\,Q(\bm{\theta},\xi).
\]

\textbf{Step 3: Equality of $J$ and $O$.}

Summing over all scenarios,
\[
  J(\bm{\theta})
  = \sum_{\xi\in\Xi}
    \min_{\hat{x}^\xi\in[0,1]} \sum_{f,k} -P(\xi)\,\hat{x}_{kf}^\xi u_{kf}^\xi
  = \sum_{\xi\in\Xi} P(\xi)\,Q(\bm{\theta},\xi)
  = O(\bm{\theta}).
\]

Therefore, for any fixed $\bm{\theta}$, RAP-AS has an optimal solution with $\hat{x}_{kf}^\xi \in\{0,1\}$ for all $(k,f,\xi)$ and achieves the same optimal objective value as the original SP. In particular, $J = O$.
\end{proof}

\section{Illustrative example of Theorem 3}
\label{seca: toy_example}

Consider one vehicle $k=1$ and two passengers $R=\{1,2\}$. 
The feasible trip set is $F=\{f_1,f_2,f_{12}\}$, where $f_1=\{1\}$, $f_2=\{2\}$, and $f_{12}=\{1,2\}$. 
Note that if the pooled trip $f_{12}$ satisfies the platform's detour and pickup time constraints, then the single-passenger trips $f_1$ and $f_2$ are necessarily feasible, as they impose strictly less detour and no additional pickup time compared to pooling.
All trips have identical utility $u_{kf}=1$.

For this scenario, the RAP-AS subproblem (relaxation) reads
\begin{align*}
\max\ & x_{1,f_1} + x_{1,f_2} + x_{1,f_{12}} \\
\text{s.t. } & x_{1,f_1} + x_{1,f_2} + x_{1,f_{12}} \le 1, \\
& x_{1,f_1} + x_{1,f_{12}} \le 1, \\
& x_{1,f_2} + x_{1,f_{12}} \le 1, \\
& x_{1,f_1}, x_{1,f_2}, x_{1,f_{12}} \in [0,1].
\end{align*}
A fractional feasible solution is $x_{1,f_1}=x_{1,f_{12}}=0.5$ and $x_{1,f_2}=0$, yielding objective value $1$. 
However, this point is \emph{not} an extreme point of the feasible polyhedron; it is the convex combination $\tfrac{1}{2}(1,0,0)+\tfrac{1}{2}(0,0,1)$ of integer solutions. 
The only extreme points are $(1,0,0)$, $(0,1,0)$, $(0,0,1)$, and $(0,0,0)$, all of which are integral and yield objective value at most $1$. 
Consequently, the LP relaxation admits an optimal integer solution (e.g., $x_{1,f_1}=1$ with objective $1$), which achieves the same optimal value as the original SP. Hence $J=O=1$.

\section{Proof of Theorem 4}

\begin{proof}
We prove Theorem~\ref{thm: scenario_reduction} by showing that any optimal solution can assign at most $|K|$ trips, and that the optimal value and assignment pattern are preserved when restricting the problem to precisely those assigned trips.

From $\sum_{f\in F^\xi} x_{kf}^\xi \le 1$ for all $k\in K$, any feasible solution satisfies
\[
    \sum_{f\in F^\xi}\sum_{k\in K} x_{kf}^\xi
    \;\le\; \sum_{k\in K} 1
    \;=\; |K|.
\]
Hence, at most $|K|$ trips can be assigned in any feasible (hence optimal) solution.

Let $x^{\xi *}$ be an optimal solution of the above problem in the scenario $\xi$. In addition, define the active set $F'$ as follows:
\[
    F' := \{ f \in F^\xi : \exists\,k\in K \text{ with } x_{kf}^{\xi *} = 1 \},
\]
where $|F'| \le |K|$. For all $f\in F^\xi\setminus F'$, we have $x_{kf}^{\xi *} = 0$ for every $k\in K$. Furthermore, consider the problem obtained by restricting trips to $F'$ and keeping the same constraints and objective, \emph{i.e.},
\[
    \min_{x}\;\sum_{f\in F'}\sum_{k\in K} -x_{kf} u_{kf}^\xi
\]
with the corresponding versions of the capacity and passenger constraints restricted to $f\in F'$. The restriction of $x^{\xi *}$ to $K\times F'$ is feasible for this restricted problem, and its objective value equals
\[
    \sum_{f\in F'}\sum_{k\in K} -x_{kf}^{\xi *} u_{kf}^\xi
    \;=\;
    \sum_{f\in F^\xi}\sum_{k\in K} -x_{kf}^{\xi *} u_{kf}^\xi,
\]
since all $x_{kf}^{\xi *}$ for $f\notin F'$ are zero. Thus, the restricted problem attains the same objective value as the original one. On the other hand, if the restricted problem had a strictly better feasible solution $\bar{x}$, then extending it to $F^\xi$ by setting $\bar{x}_{kf}=0$ for $f\notin F'$ would yield a feasible solution of the original problem with strictly smaller cost, contradicting the optimality of $x^{\xi *}$. Hence, the restricted problem on $(K,F')$ and the original problem on $(K,F^\xi)$ have the same optimal value, and the corresponding optimal assignments coincide on $K$ (with all trips in $F^\xi\setminus F'$ left unassigned). This proves the claim.
\end{proof}

\end{document}